\documentclass[11pt]{article} 

\usepackage[utf8]{inputenc} 

\usepackage{geometry} 
\usepackage{graphicx} 
\usepackage{hyperref}

\usepackage{booktabs} 
\usepackage{array} 
\usepackage{paralist} 
\usepackage{verbatim} 
\usepackage{subfig} 
\usepackage{parskip}

\usepackage{fancyhdr} 
\usepackage{sectsty}
\allsectionsfont{\sffamily\mdseries\upshape} 

\usepackage[titles,subfigure]{tocloft} 

\usepackage{amsmath, amsthm, amssymb}
\usepackage{mathrsfs}

\usepackage{chngcntr}
\counterwithin{figure}{section}

\theoremstyle{definition}
\newtheorem{definition}{Definition}[section]
\newtheorem{theorem}[definition]{Theorem}
\newtheorem{lemma}[definition]{Lemma}
\newtheorem{corollary}[definition]{Corollary}
\newtheorem{conjecture}[definition]{Conjecture}

\newcommand{\setcover}{\textsc{Set Cover}}
\newcommand{\Ell}{\mathcal{L}}

\usepackage{tikz}
\usetikzlibrary{shapes, arrows.meta, positioning}

\tikzset{%
  >={latex},
  header/.style = {rectangle, rounded corners, draw=black,
                           minimum width=3.25cm, minimum height=1cm,
                           text centered},
  construction/.style = {header, fill=cyan!25,
				 minimum width=4cm},
  element/.style = {header, rounded rectangle, fill=orange!10},
  inputconstruction/.style = {construction, fill=magenta!15},
  goalconstruction/.style = {construction, fill=green!30},
  goalelement/.style = {element, fill=orange!20},
  graphnode/.style = {circle, draw, fill=gray!30}
}

\makeatletter \g@addto@macro\@floatboxreset\centering \makeatother

\title{A Chronology of Set Cover Inapproximability Results}
\author{Erika Melder}

\begin{document}
\maketitle

\section*{Abstract}
It is well-known that an algorithm exists which approximates the NP-complete problem of Set Cover within a factor of ln(n), and it was recently proven that this approximation ratio is optimal unless P = NP. This optimality result is the product of many advances in characterizations of NP, in terms of interactive proof systems and probabilistically checkable proofs (PCP), and improvements to the analyses thereof. However, as a result, it is difficult to extract the development of Set Cover approximation bounds from the greater scope of proof system analysis. This paper attempts to present a chronological progression of results on lower-bounding the approximation ratio of Set Cover. We analyze a series of proofs of progressively better bounds and unify the results under similar terminologies and frameworks to provide an accurate comparison of proof techniques and their results. We also treat many preliminary results as black-boxes to better focus our analysis on the core reductions to Set Cover instances. The result is alternative versions of several hardness proofs, beginning with initial inapproximability results and culminating in a version of the proof that ln(n) is a tight lower bound.

\section{Introduction}

\subsection{The Set Cover Problem}

The optimization problem \setcover{} is defined as follows:

\begin{definition}[\setcover{}]
Given a set $U$ (a \emph{universe}) of $n$ elements and sets $S_1, S_2, \dots, S_k \subseteq U$, with the guarantee that every element of $U$ is in at least one of the $S_i$, find a minimal collection of subsets whose union is $U$.
\end{definition}

This problem is well-known to be NP-complete \cite{Karp-1972}. As such, the best efficient generic solution we can hope for is an approximation algorithm, which finds a solution to any instance with no more than some function of the smallest possible number of sets. The approximation version of \setcover{} is defined as follows:

\begin{definition}[\setcover{} Approximation]
For a \setcover{} instance $\mathcal{S}$, let $\text{OPT}(\mathcal{S})$ be the number of sets $S_i$ used in a minimal solution of $\mathcal{S}$. An algorithm is said to $a(n)${\it-approximate} \setcover{}, where $a$ is a function of the number of universe elements $n$ in $\mathcal{S}$, if it can generate a solution to any \setcover{} instance $\mathcal{S}$ using at most $a(n) \cdot \text{OPT}(\mathcal{S})$ sets.
\end{definition}

In particular, we are interested in approximations where $a(n)$ is minimized, and which run in polynomial time.

The lower bounds on time complexity of these approximations are well-studied. Notably, several algorithms exist which are proven to $\ln n$-approximate \setcover{} \cite{Johnson-1974, Lovasz-1975}.  On the other hand, there have been several results proving that efficient approximation is impossible within certain thresholds. The proofs of these impossibility results are nontrivial, and rely heavily on the machinery of interactive proof systems \cite{GMR-1989, BM-1988} and the subsequent framework of \emph{probabilistically checkable proofs} (PCP) \cite{Arora-1994, AS-1998}. This paper will trace key developments in proofs of \setcover{} inapproximability, beginning with the initial inapproximability result of Lund and Yannakakis \cite{LY-1994}, and progressing chronologicaly to explore what techniques were necessary to arrive at the current results. The ultimate goal will be to prove the following:

\begin{theorem}[Inapproximability of \setcover{}]
\label{thm:primary}
Unless $\text{P} = \text{NP}$, \setcover{} cannot be approximated within $c \ln n$ for any $0 < c < 1$.
\end{theorem}

\subsection{A Note on the Timeline}
\label{sec:timeline}
The papers are presented in chronological order of their results, but the years of publication listed are those of the respective journals that the papers are ultimately published in. Much of the work on this problem was done collaboratively, with initial results sparking debate and further work before those initial results were printed, and with several papers receiving revisions to reflect newly-obtained results of their successors reliant upon the initial results. For ease of reference, the papers in the main portion of this paper will be referred to by the date that they were published, though this will cause some of these papers to seemingly appear out of order. Appendix \ref{app:timeline} presents timelines of these papers arranged by rough order of initial results, in the form of a rough sketch of which results were strictly necessary to produce others.

\section{Lund and Yannakakis, 1994 \cite{LY-1994}}
\label{sec:ly}

An initial result was the following theorem:

\begin{theorem}
\label{thm:inapproxly}
Unless $\text{NP} \subseteq \text{DTIME}(n^{\text{polylog}(n)})$, \setcover{} cannot be approximated within a ratio of $c \log n$ for any positive $c < \frac{1}{4}$.
\end{theorem}

The proof of this result is nontrivial, relying on a modification of a Feige-Lov\'asz proof system for \textsc{SAT}, originally developed in \cite{FL-1992}. This section presents a simplified, descriptive version of the proof. We omit complex proofs of background information and treat them as black boxes to present only the core logic behind the key result.

\subsection{Preliminary Constructions}
In this section, we will describe two critical constructions which are necessary to complete the proof.

\subsubsection{Feige-Lov\'asz Verifier for \textsc{SAT}}
An interactive proof system is defined as follows:

\begin{definition}[Interactive Proof System]
\label{def:ips}
Fix an input size $n$. A \emph{two-prover, one-round interactive proof system} is a system which takes in an input of size $n$ and outputs \textsf{Yes} or \textsf{No}. It has two kinds of components: a \emph{verifier} and \emph{provers}.
\begin{itemize}
\item A \emph{verifier} consists of the following:
\begin{itemize}
\item A finite input alphabet $\Sigma$.
\item A finite set $R$ of random seeds of length $\mathcal{O}(\text{polylog}(n))$.
\item A finite set of queries of length $\mathcal{O}(\text{polylog}(n))$ that can be asked to each prover. In this case, there will be two provers, so the verifier will have two such sets, $Q_1$ and $Q_2$.
\item A finite set of answers of length $\mathcal{O}(\text{polylog}(n))$ that can be received from each prover. Again, there will be two such sets in this case, $A_1$ and $A_2$.
\item A polynomial-time computable function $f : \Sigma^n \times R \to Q_1 \times Q_2$, which generates a pair of queries based on a random seed and an input.
\item A polynomial-time computable function $\Pi : \Sigma^n \times R \times A_1 \times A_2 \to \{ \textsf{Yes}, \textsf{No} \}$ which determines the output based on the input, seed, and query answers.
\end{itemize}
\item The $i$th \emph{prover} in the system is a function from $Q_i$ to $A_i$.
\end{itemize}
\end{definition}

The verifier takes an input $x \in \Sigma^n$, and selects a random seed $r$ uniformly randomly from $R$. It computes $f(x, r)$ and receives a query pair $(q_1, q_2)$. It then provides these two queries to the provers, who in turn provide the answer pair $(a_1, a_2)$; this process of sending queries and receiving answers is known as \emph{message passing}, and the ``one-round'' designation of the prover refers to using only one iteration of message passing. After this, the verifier computes $\Pi (x, r, a_1, a_2)$, and accepts or rejects the input based on the result. A diagram of the process can be seen in Figure \ref{fig:verifier}.

\begin{figure}[tb]
\centering
\resizebox{\textwidth}{!}{%
\begin{tikzpicture}[node distance=3cm, font=\small, thick, >=latex]

\node[inputconstruction]
	(Input)
	{\textbf{Input} $\varphi$};

\node[element, below right of = Input, xshift=0.75cm]
	(Question1)
	{Query $q_1$};
	
\node[element, below of = Question1, yshift=-0.5cm]
	(Question2)
	{Query $q_2$};

\node[construction, below left of = Question2, xshift=-0.75cm, yshift=0.5cm]
	(Seed)
	{\textbf{Random Seed} $r$};

\node[construction, right of = Question1, yshift=1.5cm]
	(Prover1)
	{\textbf{Prover} $P_1$};

\node[construction, right of = Question2, yshift=1.5cm]
	(Prover2)
	{\textbf{Prover} $P_2$};

\node[goalelement, right of = Prover1, yshift=-1.5cm]
	(Answer1)
	{Answer $a_1$};

\node[goalelement, right of = Prover2, yshift=-1.5cm]
	(Answer2)
	{Answer $a_2$};
	
\path (Input) -- (Seed) node[midway] (FunctionF) {$f(\varphi, r)$};
\path (Answer1) -- (Answer2) node[coordinate,midway] (Dummy2) {};

\node[coordinate, right of = FunctionF, xshift=-2cm]
	(Dummy1)
	{};

\node[right of = Dummy2]
	(FunctionPi)
	{$\Pi (\varphi, r, a_1, a_2)$};

\node[goalconstruction, right of = Input, xshift = 10.2cm]
	(Output)
	{\textbf{Output} \textsf{Yes}/\textsf{No}};

\draw[->] (Input.south) -- (FunctionF);
\draw[->] (Seed.north) -- (FunctionF);
\draw (FunctionF.east)  -- (Dummy1);
\draw[->] (Dummy1) |- (Question1);
\draw[->] (Dummy1) |- (Question2);

\draw[->] (Question1.north) |- (Prover1);
\draw[->] (Question2.north) |- (Prover2);

\draw[->] (Prover1.east) -| (Answer1);
\draw[->] (Prover2.east) -| (Answer2);

\draw[->] (Answer1.east) -| (FunctionPi);
\draw[->] (Answer2.east) -| (FunctionPi);

\draw[->] (FunctionPi.east) -| (Output);

\end{tikzpicture}
}%
\caption{A standard two-prover one-round proof system.}
\label{fig:verifier}
\end{figure}
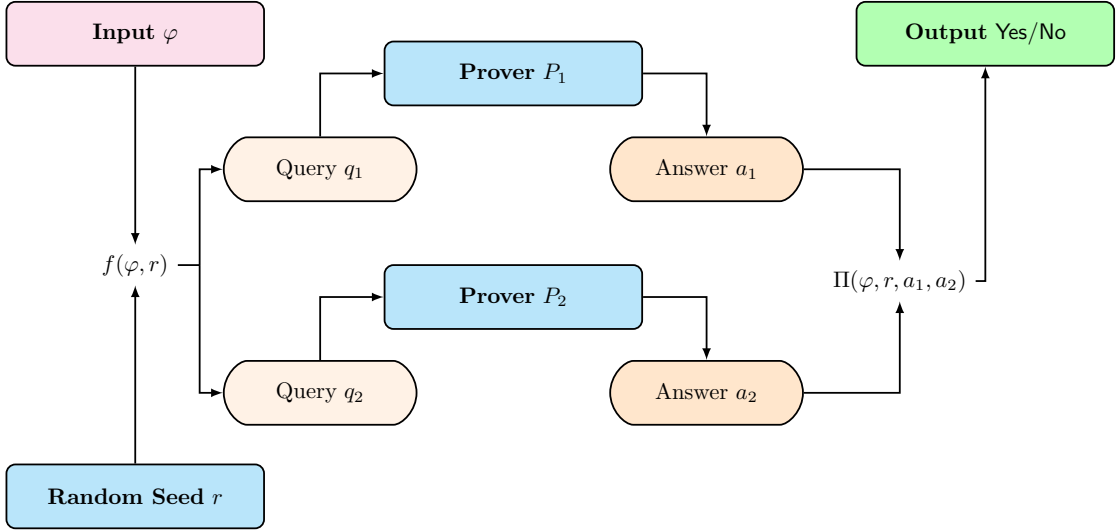

A common example of such a proof system would be a probabilistic \textsc{3SAT} verifier. Such a verifier takes in a formula $\varphi$. It selects a random clause $C$ and a random integer $i \in \{1, 2, 3\}$. The verifier asks prover 1 what assignment of variables, if any, satisfies clause $C$, and then asks prover 2 for an assignment of the $i$th variable of $C$, either 0 or 1 (corresponding to False and True). Note here that the answer sets would be $A_1 = \{ 0, 1 \}^3$ and $A_2 = \{ 0, 1 \}$. After receiving both answers, the verifier then outputs \textsf{Yes} if and only if both verifiers agree on their assignment of variable $i$. It is clear that this system is not deterministic; observe that no strategy of the provers can correctly classify both the formula $x_1 \lor x_2 \lor x_3$ and the formula $\neg x_1 \lor \neg x_2 \lor \neg x_3$ with probability 1.

Lund and Yannakakis, building on the work of Feige and Lov\'asz in \cite{FL-1992}, described a two-prover, one-round interactive proof system for \textsc{SAT} in \cite{LY-1994}, which was subsequently used as a component of the gap-preserving reduction to prove \setcover{} inapproximable. Their result is as follows:

\begin{theorem}[Proof System for \textsc{SAT}]
\label{thm:satips}
There exists a two-prover, one-round interactive proof system for SAT that takes an input formula $\varphi$ and has the following properties:
\begin{itemize}
\item If $\varphi \in \textsc{SAT}$, then there exists a pair of provers such that the verifier will always output \textsf{Yes}.
\item If $\varphi \notin \textsc{SAT}$, then for all pairs of provers, the verifier accepts with probability at most $\frac{1}{n}$, where $n$ is the input size.
\item For a fixed prover $P_i$, the verifier's queries to that prover are uniformly random across the possible query set $Q_i$.
\item The choice of $q_2$ is independent of the choice of $q_1$.
\item For any pair $(r, a_1)$ of a seed and an answer from prover 1, there is at most one $a_2 \in A_2$ that would result in a pair of accepting answers (i.e. $\Pi (\varphi, r, a_1, a_2) = \textsf{Yes}$).
\item $|Q_1| = |Q_2|$, meaning that the verifier has the same number of potential queries for each prover. (This can be assumed without loss of generality.)
\end{itemize}
\end{theorem}

We will use the existence of such a proof system as a cornerstone of the reduction. Note that the input size is fixed for verifiers of this nature, so we simply select an appropriate verifier to accommodate the size of our input. In this case, the size may be considered to be the number of literals in the formula.

\subsubsection{Special Set System}
Another critical construction is a special set system (which is a variant of an $(n, k)$-universal set as defined in \cite{NSS-1995} and described in Definition \ref{def:nkuniversalset}), defined as follows:

\begin{definition}[Special Set System]
\label{def:specialset}
A \emph{special set system} $\beta_{m, d} = \{B; C_1, C_2, \dots, C_m \}$ consists of the following:
\begin{itemize}
\item A set $B$, referred to as the \emph{universe} of the set system.
\item $m$ subsets $C_1, C_2, \dots, C_m \subseteq B$, referred to as \emph{special sets}.
\end{itemize}
Additionally, it has the following property, known as its \emph{special property}:
\begin{itemize}
\item No collection of $d$ or fewer indices $i \in \{1, 2, \dots, m\}$ satisfies $\bigcup_{i} D_i = B$, where $D_i = C_i$ or $\overline{C_i}$.
\end{itemize}
\end{definition}

Lund and Yannakakis prove the following lemma:

\begin{lemma}
For any positive integers $m$ and $d$, there exists a special set system $\beta_{m, d}$ with those parameters, and its universe $B$ has $|B| = \mathcal{O}(2^{2d}m^2)$. Moreover, such a set system can be explicitly constructed for any $m, d$ in $\text{DTIME}(n^{\text{polylog}(n)})$.
\end{lemma} 

The construction given by Lund and Yannakakis results in a set system whose elements are bit-vectors. For the purposes of the proof, the actual elements contained in the set system are irrelevant; only the special property is required.

One crucial corollary follows from taking the contrapositive of the special property:

\begin{corollary}
\label{cor:specialcor}
Let $\beta_{m, d} = \{B; C_1, C_2, \dots, C_m \}$ be a special set system, and let $A \subseteq \{C_1, \dots, C_m, \overline{C_1}, \dots, \overline{C_m}\}$ be a collection of $C_i$'s and their complements, with $|A| \leq d$. Then for some $j \in \{ 1, 2, \dots, m\}$, both $C_j$ and $\overline{C_j}$ are in $A$.
\end{corollary}

We will use this corollary to ensure the presence of the gap in our reduction.

\subsection{The Proof}

The remainder of Section \ref{sec:ly} will be devoted to proving the following lemma using the proof structure of \cite{LY-1994}, and highlighting the critical techniques therein.

\begin{lemma}
\label{lem:mainclaimly}
For a \setcover{} instance $\mathcal{S}$, let $\textsf{OPT}(\mathcal{S})$ denote the smallest number of sets needed to fully cover the universe of $\mathcal{S}$. Then, for any CNF formula $\varphi$, we can construct an instance of \setcover{} $\mathcal{S}_\varphi$ in $\text{DTIME}(n^{\text{polylog}(n)})$ such that:
\begin{itemize}
\item If $\varphi \in \textsc{SAT}$, then $\textsf{OPT}(\mathcal{S}_\varphi) = |Q_1| + |Q_2|$, where $Q_i$ is the set of queries that may be asked to prover $i$ in an appropriate Feige-Lov\'asz \textsc{SAT} proof system for $\varphi$, and
\item If $\varphi \notin \textsc{SAT}$, then $\textsf{OPT}(\mathcal{S}_\varphi) \geq c \log N (|Q_1| + |Q_2|)$ for any $0 < c < \frac{1}{4}$.
\end{itemize}
\end{lemma}

This would complete a gap-inducing reduction from \textsc{SAT} to \setcover{}. Because \textsc{SAT} is NP-complete, this would prove that \setcover{} is inapproximable within $c \log(n) \cdot \textsf{OPT}$ for any $0 < c < \frac{1}{4}$, unless $\text{NP} \subseteq \text{DTIME}(n^{\text{polylog}(n)})$. The second possibility is because the special set system construction can be done in $\text{DTIME}(n^{\text{polylog}(n)})$ as shown in \cite{LY-1994}. This would subsequently prove Theorem \ref{thm:inapproxly}.

\subsubsection{Initial Setup}
Consider an input formula $\varphi$ of length $n$, and an appropriate Feige-Lov\'asz verifier for \textsc{SAT}. Recall that the verifier consists of sets $R, Q_1, Q_2, A_1,$ and $A_2$, and functions $f$ and $\Pi$, as in Definition \ref{def:ips}.

Construct the special set system $\beta_{m, d}$, where $m = |A_2|$ and $d$ is an even integer to be determined later. Since $m = |A_2|$, we can associate each special set $C_a$ with a corresponding answer $a \in A_2$. Define $S = R \times B$, where $R$ is the random seed set of the verifier and $B$ is the universe of $\beta_{m, l}$. The set $S$ will be the universe for our \setcover{} instance.

We next add one subset to our \setcover{} instance $\mathcal{S}_\varphi$ for each query-answer pair in $Q_1 \times A_1$ and $Q_2 \times A_2$. Note that ultimately the inapproximability result is independent of the number of subsets in the instance, but the actual construction of the appropriate verifier is the source of the number of subsets.  Let $q(r, i)$ be the query that the verifier asks prover $i$ when using seed $r$. Because $\varphi$ is an input to the query-generating function $f$, determining this query depends on $\varphi$. Let $a_2(r, a_1)$ be the answer that prover 2 can give such that $\Pi(\varphi, r, a_1, a_2(r, a_1))$ accepts, or \textsf{Null} if no such answer exists. Recall that Theorem \ref{thm:satips} guarantees that $a_2(r, a_1)$ is unique if it exists.

The sets corresponding to $Q_1 \times A_1$ are given by
$$ S_{q_1, a_1} = \{ (r, b) \in R \times B : q_1 = q(r, 1), \text{and } a_2(r, a_1) \text{ is non-\textsf{Null}, and } b \notin C_{a_2(r, a_1)}\} .$$

This is the set of all seed-element pairs in $R \times B$ where:
\begin{itemize}
\item $q_1$ is the question asked to prover 1 for input $\varphi$ and seed $r$;
\item There is an answer from prover 2 $a_2(r, a_1)$ that causes the verifier to accept when $a_1$ is the answer from prover 1; and
\item $b$ is \emph{not} in the special set corresponding to $a_2(r, a_1)$.
\end{itemize}

The sets corresponding to $Q_2 \times A_2$ are defined similarly:
$$ S_{q_2, a_2} = \{ (r, b) : q_2 = q(r, 2) \text{ and } b \in C_{a_2}\}.$$

This is the set of all seed-element pairs in $R \times B$ where:
\begin{itemize}
\item $q_2$ is the question asked to prover 2 for input $\varphi$ and seed $r$; and
\item $b$ \emph{is} in the special set corresponding to $a_2$.
\end{itemize}

Then the \setcover{} instance has universe set $U$ and all $S_{q_i, a_i}$ for $i = 1, 2$ as subsets. 
One important lemma arises from this construction. A diagram of the construction dependencies for the instance can be seen in Figure \ref{fig:construction}.

\begin{figure}[tb!]
\centering
\begin{tikzpicture}[node distance=1.5cm, font=\small,thick, >=Latex]

\node[goalconstruction]
	(SetCover)
	{\textbf{Set Cover Instance}};

\node[goalelement, below of = SetCover]
	(SetCoverU)
	{Universe Set};

\node[goalelement, below of = SetCoverU]
	(SetCoverS)
	{Subsets};

\node[construction, below left of = SetCover, xshift=-4cm] 
	(Verifier)
	{\textbf{Verifier for SAT}};

\node[element, below of = Verifier]
	(VerifierR)
	{Random Seeds};

\node[element, below of = VerifierR]
	(VerifierQ1)
	{Questions for $P_1$};

\node[element, below of = VerifierQ1]
	(VerifierQ2)
	{Questions for $P_2$};

\node[element, below of = VerifierQ2]
	(VerifierA1)
	{Answers from $P_1$};

\node[element, below of = VerifierA1]
	(VerifierA2)
	{Answers from $P_2$};

\node[construction, below right of = SetCover, xshift=4cm]
	(SpecialSets)
	{\textbf{Special Set System}};

\node[element, below of = SpecialSets]
	(SpecialSetsB)
	{Universe Set};

\node[element, below of = SpecialSetsB]
	(SpecialSetsC)
	{Subsets};

\path (VerifierQ2) -- (VerifierA1) node[coordinate,midway] (Dummy1Helper) {};

\node [coordinate, right of = Dummy1Helper, xshift = 1cm] (Dummy1) {};
\node [coordinate, right of = VerifierR, xshift = 1cm] (Dummy2) {};
\node [coordinate, left of = SpecialSetsB, xshift = -1cm] (Dummy3) {};

\draw[->] (VerifierA2) -| (SpecialSetsC);

\draw (VerifierQ1.east) -- (Dummy1);
\draw (VerifierQ2.east) -- (Dummy1);
\draw (VerifierA1.east) -- (Dummy1);
\draw (VerifierA2.east) -- (Dummy1);
\draw[->] (Dummy1) -| (SetCoverS);

\draw (SpecialSetsC.west) -| (SetCoverS);

\draw (VerifierR.east) -- (Dummy2);
\draw[->] (Dummy2) |- (SetCoverU);

\draw (SpecialSetsB.west) -- (Dummy3);
\draw[->] (Dummy3) |- (SetCoverU);

\end{tikzpicture}
\caption{A diagram of Lund and Yannakakis's proof structure.}
\label{fig:construction}
\end{figure}
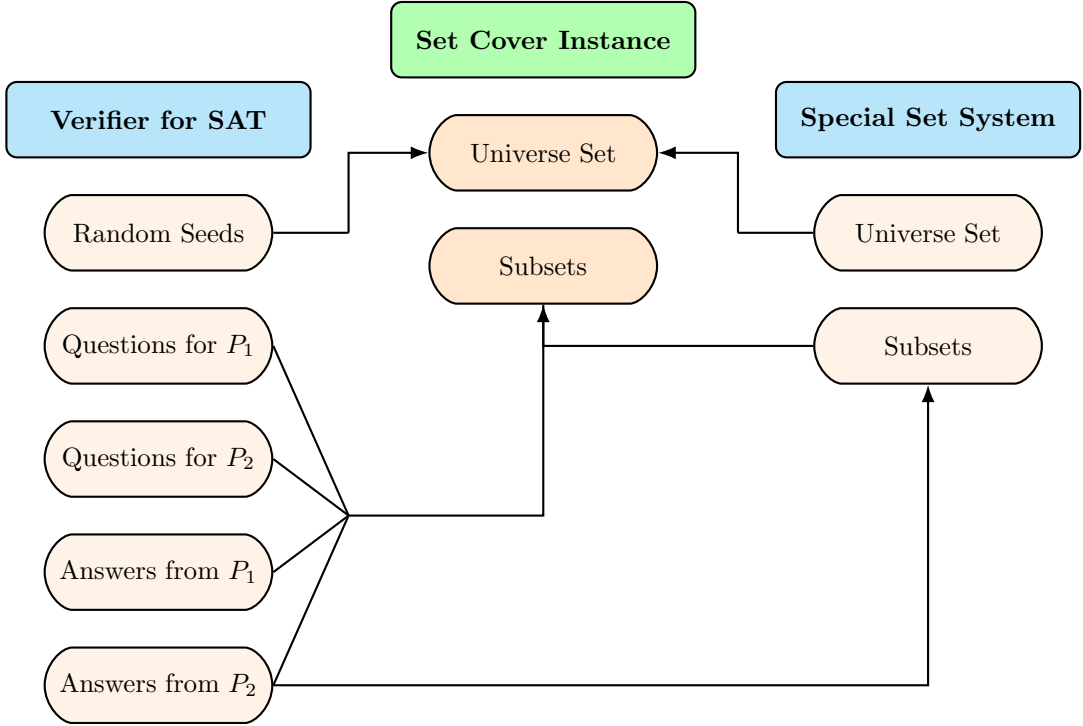

\begin{lemma}
\label{lem:coverseed}
For a fixed seed $r \in R$, we can cover all elements $\{(r, b)\}_{b \in B}$ with any two sets $S_{q(r, 1), a_1}$ and $S_{q(r, 2), a_2(r, a_1)}$.
\end{lemma}

\begin{proof}[Proof of Lemma \ref{lem:coverseed}]
Fix a seed $r$. Assume $a_2(r, a_1)$ exists and consider the two sets $S_{q(r, 1), a_1}$ and $S_{q(r, 2), a_2(r, a_1)}$. Take a pair $(r, b)$. Then there are two cases:
\begin{itemize}
\item Case 1: $(r, b) \in S_{q(r, 1), a_1}$. Done.
\item Case 2: $(r, b) \notin S_{q(r, 1), a_1}$. Since $q(r, 1)$ is the question asked to prover 1 for seed $r$ and $a_2(r, a_1)$ exists, the only way this is possible is if $b \in C_{a_2(r, a_1)}$. But then, since $q(r, 2)$ is the question asked to prover 2 for seed $r$, and $b \in C_{a_2(r, a_1)}$, it follows that $(r, b) \in S_{q(r, 2), a_2(r, a_1)}$ by construction.
\end{itemize}
Since all $(r, b)$ must be in one of the two sets, their union must cover all of $\{(r, b)\}_{b \in B}$.
\end{proof}

An important corollary of Lemma \ref{lem:coverseed} arises, which allows us to conclude bounds on the sizes of optimal and suboptimal coverings and thus propagate the gap. The statement of the corollary is as follows:

\begin{corollary}
\label{cor:setcontain}
For a fixed seed $r \in R$, any covering which uses $d$ or fewer of the sets $S_{q_i, a_i}$ to cover $\{(r, b)\}_{b \in B}$ must contain two sets of the form $S_{q(r, 1), a_1}$ and $S_{q(r, 2), a_2(r, a_1)}$.
\end{corollary}

\begin{proof}[Proof of Corollary \ref{cor:setcontain}]
Assume not. Notice that sets $S_{q(r, 1), a_1}$ cover all pairs $(r, b)$ for which $b \in \overline{C_{a_2(r, a_1)}}$, which is the complement of a special set. Similarly, the sets $S_{q(r, 2), a_2}$ cover all pairs $(r, b)$ for which $b \in C_{a_2}$, which is a special set. By assumption, there is no pair of sets $S_{q(r, 1), a_1}$ and $S_{q(r, 2), a_2(r, a_1)}$ in the covering, so none of the sets $\overline{C_{a_2(r, a_1)}}$ are complements of any of the sets $C_{a_2}$.

However, this means that a collection of $d$ or fewer of the special sets $C_a$ and their complements $\overline{C_a}$ is able to cover all $b \in B$, since there is one pair $(r, b)$ for each $b \in B$. By Corollary \ref{cor:specialcor}, some set and its own complement must be in the collection, which contradicts the assumption. Therefore, a pair of sets of the form $S_{q(r, 1), a_1}$ and $S_{q(r, 2), a_2(r, a_1)}$ must be in the covering.
\end{proof}

With this knowledge, we can reformulate Lemma \ref{lem:mainclaimly}:

\begin{lemma}
\label{lem:refclaim}
Let $\mathcal{S}$ be a \setcover{} instance. Then, for any CNF formula $\varphi$, we can construct an instance of \setcover{} $\mathcal{S}_\varphi$ in $\text{DTIME}(n^{\text{polylog}(n)})$ such that:
\begin{itemize}
\item If $\varphi \in \textsc{SAT}$, then $\textsf{OPT}(\mathcal{S}_\varphi) = |Q_1| + |Q_2|$.
\item If $\varphi \notin \textsc{SAT}$, then $\textsf{OPT}(\mathcal{S}_\varphi) \geq c \log (|R \times B|) \cdot (|Q_1| + |Q_2|)$ for any $0 < c < \frac{1}{4}$.
\end{itemize}
\end{lemma}

We will prove each part of Lemma \ref{lem:refclaim} separately.

\subsubsection{Proving the First Part}

The proof of the first part is straightforward. We seek to prove that if $\varphi \in \textsc{SAT}$, then $\textsf{OPT}(\mathcal{S}_\varphi) = |Q_1| + |Q_2|$.

Assume that $\varphi \in \textsc{SAT}$. Then there exists some pair of provers which makes the verifier accept with probability 1. This means that for all seeds, if seed $r$ generates query pair $(q_1, q_2)$, then the answer pair $(a_1, a_2)$ returned by the provers causes $\Pi(\varphi, r, a_1, a_2)$ to output \textsf{Yes}. Notice that this means the sets $S_{q(r, 1), a_1}$ and $S_{q(r, 2), a_2(r, a_1)}$ cover all $\{(r, b)\}_{b \in B}$, and such a pair of sets therefore exists in our instance for every seed $r$.
Additionally, because the prover's strategies are fixed, there is exactly one set per query in the strategy, meaning there are $|Q_1| + |Q_2|$ corresponding sets in the instance. Since the sets in the strategy are able to cover all pairs, it follows that the size of a minimal covering for $S$ is $|Q_1| + |Q_2|$. This proves the first part.

\subsubsection{Proving the Second Part}

The proof of the second part is more involved. We seek to prove that if $\varphi \notin \textsc{SAT}$, then $\textsf{OPT}(\mathcal{S}_\varphi) \geq c \log (|R \times B|) \cdot (|Q_1| + |Q_2|)$ for any $0 < c < \frac{1}{4}$.

Assume that $\varphi \notin \textsc{SAT}$. Define a bipartite graph $G$, and let the two sets of nodes be $Q_1$ and $Q_2$. Add an edge $e_r$ between queries $q_1 \in Q_1$ and $q_2 \in Q_2$ if there is some random seed $r$ such that $f(\varphi, r) = (q_1, q_2)$. By the properties of the proof system, and because $|Q_1| = |Q_2|$, $G$ is $\frac{|R|}{|Q_1|}$-regular.

Consider some arbitrary covering $\mathcal{C}$ of $S$. We will use this to define a cost function $c$ as follows:
\begin{itemize}
\item The cost of a node $c(q)$ is the number of sets $S_{q, a}$ in $\mathcal{C}$ for which the query is $q$.
\item The cost of an edge $c(e_r)$, where $e_r = (q_1, q_2)$, is the sum of the costs of its incident nodes, $c(q_1) + c(q_2)$.
\end{itemize}

Note that if the cost of edge $e_r$ is $\leq d$, then because $\mathcal{C}$ is a covering, it follows that all pairs $\{(r, b)\}_{b \in B}$ can be covered in fewer than $d$ sets. By Corollary $\ref{cor:setcontain}$, it follows that seed $r$ has a pair of satisfying answers to its queries.

Let an edge $e_r$ be called \emph{good} if $c(e_r) \leq d$, and let it be called \emph{bad} otherwise. Let $\delta$ be the fraction of good edges in $G$. We then prove two lemmas:

\begin{lemma}
\label{lem:probability}
If $\varphi \notin \textsc{SAT}$, there exists provers such that the verifier accepts with probability at least $\frac{\delta}{d^2}$.
\end{lemma}

\begin{proof}[Proof of Lemma \ref{lem:probability}]
Define $d$ strategies $P_{i, 1}, P_{i, 2}, \dots, P_{i, d}$ for each of the two provers $P_i$ as such:
\begin{itemize}
\item For each query $q \in Q_i$, look at the strategies $a \in A_i : S_{q, a} \in \mathcal{C}$. Order the strategies arbitrarily.
\item The answer $P_{i, j}$ gives to query $q$ will be the $j$th element of the ordering of $q$'s answers, or an arbitrary answer if there are no more in the ordering.
\end{itemize}
If a seed $r$ is chosen uniformly randomly, it has probability $\delta$ of corresponding to a good edge, meaning that there is an accepting pair of answers for $r$'s queries. These answers must belong to some pair of strategies of the provers by pigeonhole principle, since there are at most $d$ question-answer pairs corresponding to each query on a good edge, and each of the $d$ strategies provides a different answer out of those in the covering. This means that some pair of strategies has at least a $\frac{\delta}{d^2}$ chance of making the verifier accept: probability $\delta$ that the edge $e_r$ is good, and probability at least $\frac{1}{d^2}$ that our strategies choose an accepting answer pair in response.
\end{proof}

\begin{lemma}
\label{lem:coversize}
The size of the covering set, $|\mathcal{C}|$, is at least $(1 - \delta) \cdot \frac{d}{2} \cdot (|Q_1| + |Q_2|)$.
\end{lemma}

\begin{proof}[Proof of Lemma \ref{lem:coversize}]
The sum of all edge costs is bounded below by the sum of all bad edge costs. Since each bad edge costs $>d$, and the number of bad edges is $(1 - \delta) \cdot |R|$, we can conclude that the sum of edge costs is strictly bounded below by $(1 - \delta) \cdot |R| \cdot d$.

The sum of all edge costs is also equal to the sum of all costs of the corresponding vertices to each edge. Since $G$ is $\frac{|R|}{|Q_1}$-regular, this evaluates to $\frac{|R|}{|Q_1|} \cdot \sum_{q \in Q_1 \cup Q_2} c(q)$. By construction, the sum of all individual vertex costs is $C$.

Combining these two facts yields
\begin{align*}
\frac{|R|}{|Q_1|} |\mathcal{C}| &> (1 - \delta) |R| d \\
|\mathcal{C}| &> (1 - \delta) \cdot d \cdot |Q_1| \\
|\mathcal{C}| &> (1 - \delta) \cdot \frac{d}{2} \cdot 2|Q_1| \\
|\mathcal{C}| &> (1 - \delta) \cdot \frac{d}{2} \cdot (|Q_1| + |Q_2|) \\
\end{align*} 
since $|Q_1| = |Q_2|$.
\end{proof}

Now we can complete the proof. By Lemma \ref{lem:probability}, if $\varphi \notin \textsc{SAT}$, then $\frac{\delta}{d^2} \leq \frac{1}{n}$. Substituting directly into the result of Lemma \ref{lem:coversize} yields that
$$|\mathcal{C}| \geq \left( \frac{d}{2} - \frac{d^3}{2n} \right)(|Q_1| + |Q_2|).$$

Choose arbitrarily large $k > 0$ and let $d = \max (k \log |R|, k \log|A_2|)$. This choice of $d$ ensures that the following step will hold.

Note that $|R \times B| = |R||B|$, and by construction, $|B| = \mathcal{O}(2^{2d}|A_2|^2)$. This means that
$$\log(|R||B|) \leq 2d \left( 1 + \frac{3}{2k} \right).$$

We may make $k$ arbitrarily large, which makes the $\frac{3}{2k}$ vanish as $k$ increases. This results in the simplified statement
$$\log(|R||B|) \leq 2 d$$
Solving for $d$, substituting into the result of Lemma \ref{lem:coversize}, and simplifying yields
$$|\mathcal{C}| \geq \frac{1}{2} \cdot \left( \frac{\log(|R||B|)}{2} - \frac{\log(|R||B|)^3}{8|R||B|} \right) \cdot (|Q_1| + |Q_2|).$$
We may assume that $\log(|R||B|)^3 \ll 8|R||B|$, which yields the final result of
$$|\mathcal{C}| \geq \frac{1}{4} \cdot \log(|R||B|) \cdot (|Q_1| + |Q_2|)$$
which proves the second part of the claim. This proves Lemma \ref{lem:refclaim}, hence proving Lemma \ref{lem:mainclaimly} and completing the reduction.

\subsection{Using a Randomized Construction}

By using a randomized construction for the special set system, rather than a deterministic one, Lund and Yannakakis are able to construct it in $\text{ZTIME}(n^{\text{polylog}(n)})$, where ZTIME is zero-error probabilistic polynomial time (the class of languages decidable via Las Vegas algorithms). The size of the new special set is $|B| = (d + d \ln m + 2)2^d$, which is approximately the square root of the size of the deterministic algorithm, which is $|B| = \mathcal{O}(2^{2 d} m^2)$. Using this, they achieve the following result.

\begin{theorem}
\label{thm:inapproxly-rand}
Unless $\text{NP} \subseteq \text{ZTIME}(n^{\text{polylog}(n)})$, \setcover{} cannot be approximated within a ratio of $c \log n$ for any positive $c < \frac{1}{2}$.
\end{theorem}

This is a stronger result than using the deterministic construction, but requires a stronger assumption in turn. As such, it is effectively incomparable to the deterministic result in terms of quality.

\section{Bellare, Goldwasser, Lund, and Russell, 1993 \cite{BGLR-1993}}
\label{sec:bglr}

\subsection{Inapproximability Result}

The main theorem that this paper achieves is as follows:

\begin{theorem}
\label{thm:inapproxbglr} ~
\begin{enumerate}
\item Unless $\text{P} = \text{NP}$, \setcover{} cannot be approximated within a ratio of $c$ for any constant $c$.
\item Unless $\text{NP} \subseteq \text{DTIME}(n^{\mathcal{O}(\log \log n)})$, \setcover{} cannot be approximated within a ratio of $c \log n$ for any positive $c < \frac{1}{8}$.
\end{enumerate}
\end{theorem}

These results come from increasing the number of provers in the system to four, and providing a slightly different construction. Though the assumptions of these results are weaker, the constants are also worse; therefore, this result is effectively incomparable with the result of Section \ref{sec:ly}.

\subsection{The Sliding Scale Conjecture}

The paper also formulates the following conjecture, a weaker version of which was ultimately proven in \cite{DS-2014}.

\begin{conjecture}[Sliding Scale Conjecture]
\label{conj:slidingscale}
There exist two-prover, one-round proof systems for SAT that have logarithmic randomness, logarithmic answer sizes, and $\frac{1}{n}$ error probability.
\end{conjecture}

Contingent on this conjecture, the authors posit that better \setcover{} results are achievable. In particular, they observe that a \setcover{} result arises from a weaker form of the conjecture.

\begin{conjecture}[Weak Sliding Scale Conjecture for \setcover{}]
Suppose that for some constant $p$ and some function $\epsilon(n) = \frac{1}{\log^{\omega(1)} n}$, SAT has $p$-prover proof systems with logarithmic randomness, logarithmic answer sizes, and error $\epsilon(n)$. Then, unless $\text{P} = \text{NP}$, \setcover{} cannot be approximated within a ratio of $c \log n$ for any positive $c < \frac{1}{2p}$.
\end{conjecture}

While this conjecture remains unproven, it forms the basis of the Projection Games Conjecture in \cite{Moshkovitz-2015}.

\section{Raz, 1998 \cite{Raz-1998}}

Raz does not focus on \setcover{} in this paper; rather, he offers a breakthrough in error of multi-prover interactive proof systems in the form of the following lemma:

\begin{lemma}[Raz Repetition Lemma]
\label{lem:razrepetition}
If a two-prover, one-round proof system is repeated $\ell$ times independently in parallel, then the error is $2^{-c\ell}$, where $c > 0$ is constant and is dependent only on the error in the original proof system and the length of the answers in that proof system.
\end{lemma}

This implies that parallel repetition of a two-prover one-round proof system reduces the error of the system exponentially fast. Taken in combination with Theorems \ref{thm:inapproxly} and \ref{thm:inapproxly-rand}, the hardness assumptions can be relaxed, yielding the following:

\begin{theorem}
\label{thm:inapproxraz} ~
\begin{enumerate}
\item Unless $\text{NP} \subseteq \text{DTIME}(n^{\mathcal{O}(\log \log n)})$, \setcover{} cannot be approximated within a ratio of $c \log n$ for any positive $c < \frac{1}{4}$.
\item Unless $\text{NP} \subseteq \text{ZTIME}(n^{\mathcal{O}(\log \log n)})$, \setcover{} cannot be approximated within a ratio of $c \log n$ for any positive $c < \frac{1}{2}$.
\end{enumerate}
\end{theorem}

Both parts of this theorem are strict improvements on the results of \cite{LY-1994}; the first part is also a strict improvement over the second part of Theorem \ref{thm:inapproxbglr} from \cite{BGLR-1993}.

\section{Naor, Schulman, and Srinivasan, 1995 \cite{NSS-1995}}

\subsection{\protect$(n, k)$-Universal Sets}

The key component of this paper was the development of new techniques for constructing $(n, k)$-universal sets, which are defined as follows:

\begin{definition}[$(n, k)$-Universal Set]
\label{def:nkuniversalset}
An $(n, k)$-universal set $T \subseteq \{ 0, 1 \}^n$ is a minimal set of bitstrings such that, for any set of indices $S \subseteq [n]$ with $|S| = k$, the projection of $T$ on $S$ (defined as the set of subsequences of bitstrings of $T$ taken at the indices in $S$) is exactly the set $\{ 0, 1 \}^k$. 
\end{definition}

Naor, Schulman, and Srinivasan developed new deterministic procedures for constructing these sets in \cite{NSS-1995}. Their main result is the following:

\begin{theorem}
There is a deterministic, explicit construction of $(n, k)$-universal sets of size $2^k k^{\mathcal{O}(\log k)} \log n$.
\end{theorem}

The theorem is a corollary of their more general result on $k$-restriction problems. Their procedure makes the deterministic construction much more efficient, yielding performance on par with the randomized construction. This unifies the deterministic and randomized lower bounds, yielding the following refinement of Theorem \ref{thm:inapproxraz}:

\begin{theorem}
\label{thm:inapproxnss}
Unless $\text{NP} \subseteq \text{DTIME}(n^{\mathcal{O}(\log \log n)})$, \setcover{} cannot be approximated within a ratio of $c \log n$ for any positive $c < \frac{1}{2}$.
\end{theorem}

\subsection{\protect$(n, k, b)$-Anti-Universal Sets}

In addition to $(n, k)$-universal sets, another result of Naor \emph{et al.} on restriction games involves a structure known as $(n, k, b)$-anti-universal sets, which were initially proposed by Feige and are a cornerstone of the reduction in \cite{Feige-1998}. They are defined as follows:

\begin{definition}[$(n, k, b)$-Anti-Universal Set]
\label{def:nkbantiuniversalset}
An $(n, k, b)$-anti-universal set $\mathcal{T} \subseteq [n] \times [b]$ is a family of functions $t_i : [n] -> [b]$ where, for every pair of vectors $u \in [n]^k$ and $v \in [b]^k$, there is a function $t_j \in \mathcal{T}$ which transforms $u$ into a vector that disagrees with $v$ on each coordinate (meaning that $t_j$ takes every element $u_i \in u$ to something \emph{other} than the corresponding element $v_i \in v$). 
\end{definition}

This is a generalization of $(n, k)$-universal sets, which are simply the special case where $b = 2$. Naor \emph{et al.} achieve the following result:

\begin{theorem}
\label{thm:nkbantiuniversalsetconstruction}
For any fixed $b$, there is a deterministic, explicit construction of $(n, k, b)$-anti-universal sets of size $\left( \frac{b}{b - 1} \right)^k k^{\mathcal{O}(\log k)} \log n$.
\end{theorem}

In combination with the reduction of Section \ref{sec:feige}, this is sufficient to prove Theorem \ref{thm:inapproxfeige}.

\section{Feige, 1998 \cite{Feige-1998}}
\label{sec:feige}

Feige provides the first crucial innovation by generalizing the frameworks of \cite{LY-1994} and by beginning from MAX 3SAT-5 rather than conventional SAT. The ultimate result is as follows:

\begin{theorem}
\label{thm:inapproxfeige}
Unless $\text{NP} \subseteq \text{DTIME}(n^{\mathcal{O}(\log \log n)})$, \setcover{} cannot be approximated within a ratio of $c \ln n$ for any positive $c < 1$.
\end{theorem}

We will proceed to give a proof of this theorem following the constructions given by Feige.

\subsection{Preliminary Constructions}

In this section, we describe several constructions and concepts necessary for Feige's proof, many of which are generalizations of constructions in Section \ref{sec:ly}.

\subsubsection{Interactive Proof System for MAX 3SAT-5}

The problem of MAX 3SAT-5 is a particular member of a family of SAT problems with restrictions on the number of clauses each variable appears in. The problem statement is:

\begin{definition}[MAX 3SAT-5]
Consider a CNF boolean formula $\varphi$ on $n$ variables and $\frac{5}{3}n$ clauses, where the following three conditions hold (such formulas are known as 3SAT-5 formulas):
\begin{itemize}
\item Each clause contains 3 literals.
\item Each variable appears in exactly 5 clauses.
\item No variable appears in the same clause more than once.
\end{itemize}
The problem of MAX 3SAT-5 is to determine the maximum number of clauses that are simultaneously satisfiable.
\end{definition}

The following is a theorem of Feige:

\begin{theorem}
\label{thm:max3sat5distinguishability}
For some $\varepsilon > 0$, it is NP-hard to distinguish between 3SAT-5 formulas that are completely satisfiable, and those with at most a $(1 - \varepsilon)$ fraction of clauses simultaneously satisfiable.
\end{theorem}

It is this gap in distinguishability that will ultimately lead to the gap reduction.

Feige also presents an explicit protocol for a two-prover system for MAX 3SAT-5. (See Definition \ref{def:ips} for the structure of such a system.) This is different from the methodology of Lund and Yannakakis, who used a theorem asserting the existence of such a protocol without actually describing it.

A two-prover, one-round protocol for MAX 3SAT-5 is as follows:
\begin{itemize}
\item The first prover is a function $P_1 : \left[ \frac{5n}{3} \right] \to \{0, 1\}^3$. This can be thought of as a function receiving the index of some clause in $\varphi$, and returning a prospective assignment to the three variables in the clause in order.
\item The second prover is a function $P_2 : \{ 1, 2, 3 \} \to \{ 0, 1 \}$. This can be thought of as a function being told to look at the first, second, or third variable in the clause that $P_1$ received, and giving that variable a prospective assignment.
\item The verifier accepts if both of the following are true:
\begin{itemize}
\item The assignment given by $P_1$ satisfies the clause.
\item The assignment given by $P_2$ agrees with the assignment that $P_1$ gave the variable in question.
\end{itemize}
\end{itemize}

Analysis of this proof system results in the following theorem:

\begin{theorem}
Let $\varphi$ be a 3CNF-5 formula, and let $\varepsilon$ be the fraction of unsatisfied clauses in the assignment that satisfies the \emph{most} possible clauses. Then, under the optimal strategy of the provers, the verifier accepts with probability $1 - \frac{\varepsilon}{3}$.
\end{theorem}

Note that the error in this proof system is still high. To reduce the error, the system can be run in parallel: choose $\ell$ clauses and one variable from each clause, and then have the provers provide assignments to everything at once. By Lemma \ref{lem:razrepetition}, this reduces the error to $2^{-c\ell}$ for some positive constant $c$ dependent only on the original error and the answer lengths; because the original error is independent of $n$ and the answer lengths are fixed, the new error is thus $2^{-c \ell}$ for a universal constant $c$.

\subsubsection{Generalization to \protect$k$ Provers}

Let $k$ be an arbitrarily large constant. Consider a binary code containing $k$ code words of length $\rho = \Theta (\log \log n)$, weight $\frac{\rho}{2}$, and pairwise Hamming distance at least $\frac{\rho}{3}$. Such a code can be used to construct a $k$-prover, one-round protocol for MAX 3SAT-5:
\begin{itemize}
\item The verifier selects $\rho$ clauses $C_1, \dots, C_\rho$ uniformly randomly, and selects one \emph{distinguished variable} in each clause $x_1, \dots, x_\rho$ uniformly randomly.

\item Associate each prover $P_i$ with a code word of length $\rho$. Because the code word has weight $\frac{\rho}{2}$, half of its bits are 0 and half are 1. Consider bit $j$ of the code word.
\begin{itemize}
\item If bit $j$ is 0, prover $P_i$ is sent the distinguished variable $x_j$.
\item If bit $j$ is 1, prover $P_i$ is sent the clause $C_j$.
\end{itemize}
In this way, each prover receives $\frac{\rho}{2}$ clauses and $\frac{\rho}{2}$ distinguished variables to assign, and does not receive any distinguished variables from the clauses it is assigning.

\item The answer each prover $P_i$ provides is a length $2\rho$ bitstring consisting of assignments of all variables sent to $P_i$: the $\frac{\rho}{2}$ distinguished variables, and the $\frac{3 \rho}{2}$ variables in the clauses.
\end{itemize}

We say that provers $P_i$ and $P_j$ are \emph{consistent} if their assignments to the distinguished variables agree on each coordinate of their respective code words. It does not matter whether $P_i$ and $P_j$ agree on non-distinguished variables of clauses, or whether they are each internally consistent on assigning the distinguished variable the same truth value in different clauses.

To establish a clear gap, Feige introduces the notion of weak and strong acceptance of a given formula: the system \emph{weakly accepts} a formula if \emph{some} pair of provers is consistent, and it \emph{strongly accepts} a formula if \emph{every} pair of provers is consistent. A diagram of the $k$-prover extension can be seen in \ref{fig:verifier-kprover}.

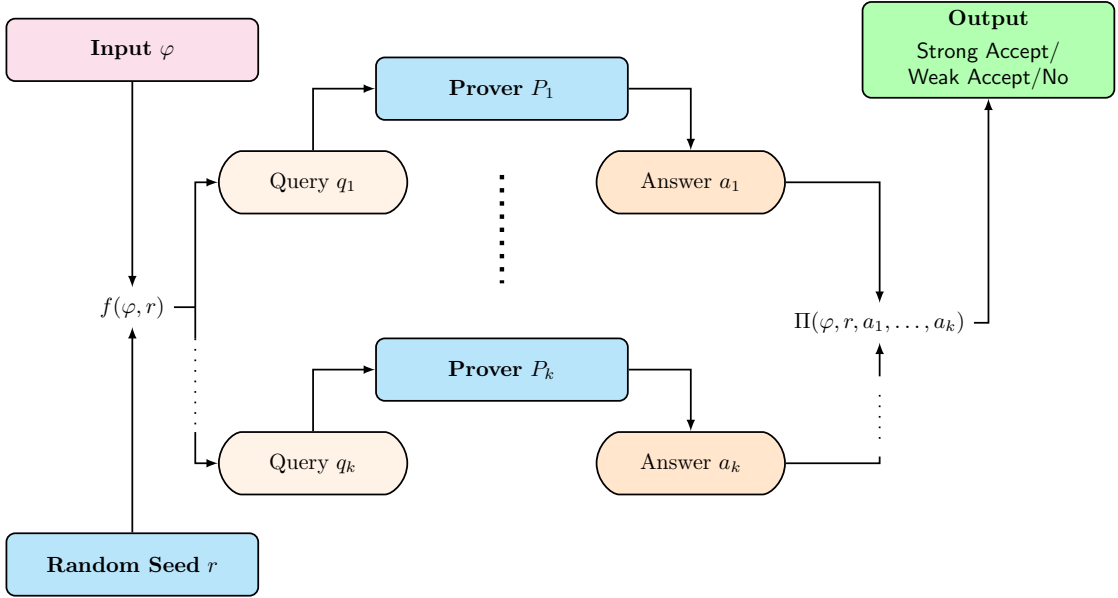
\begin{figure}[tb]
\centering
\resizebox{\textwidth}{!}{%
\begin{tikzpicture}[node distance=3cm, font=\small, thick, >=latex]

\node[inputconstruction]
	(Input)
	{\textbf{Input} $\varphi$};

\node[element, below right of = Input, xshift=0.75cm]
	(Question1)
	{Query $q_1$};
	
\node[element, below of = Question1, yshift=-1.5cm]
	(QuestionK)
	{Query $q_k$};

\node[construction, below left of = QuestionK, xshift=-0.75cm, yshift=0.5cm]
	(Seed)
	{\textbf{Random Seed} $r$};

\node[construction, right of = Question1, yshift=1.5cm]
	(Prover1)
	{\textbf{Prover} $P_1$};

\node[construction, right of = QuestionK, yshift=1.5cm]
	(ProverK)
	{\textbf{Prover} $P_k$};

\node[goalelement, right of = Prover1, yshift=-1.5cm]
	(Answer1)
	{Answer $a_1$};

\node[goalelement, right of = ProverK, yshift=-1.5cm]
	(AnswerK)
	{Answer $a_k$};
	
\path (Input) -- (Seed) node[midway] (FunctionF) {$f(\varphi, r)$};
\path (Answer1) -- (AnswerK) node[coordinate,midway] (Dummy2) {};
\path (Prover1) -- (ProverK) node[coordinate,midway] (MidDotsCenter) {};
\path (Prover1) -- (MidDotsCenter) node[coordinate,midway] (MidDotsBegin) {};
\path (MidDotsCenter) -- (ProverK) node[coordinate,midway] (MidDotsEnd) {};

\node[coordinate, right of = FunctionF, xshift=-2cm]
	(Dummy1)
	{};

\node[coordinate, below of = Dummy1, yshift=2.5cm]
	(DotsBeginLeft)
	{};

\node[coordinate, below of = DotsBeginLeft, yshift=1.5cm]
	(DotsEndLeft)
	{};

\node[right of = Dummy2]
	(FunctionPi)
	{$\Pi (\varphi, r, a_1, \dots, a_k)$};

\node[below of = FunctionPi, yshift=2cm]
	(DotsEndRight)
	{};

\node[below of = DotsEndRight, yshift=2cm]
	(DotsBeginRight)
	{};

\node[goalconstruction, right of = Input, xshift = 10.6cm]
	(Output)
	{\parbox{3cm}{\centering \textbf{Output} \\[0.3em]\textsf{Strong Accept}/\\\textsf{Weak Accept}/\textsf{No}}};

\draw[->] (Input.south) -- (FunctionF);
\draw[->] (Seed.north) -- (FunctionF);
\draw (FunctionF.east)  -- (Dummy1);
\draw[->] (Dummy1) |- (Question1);
\draw (Dummy1) -- (DotsBeginLeft);
\draw[loosely dotted] (DotsBeginLeft) -- (DotsEndLeft);
\draw[->] (DotsEndLeft) |- (QuestionK);

\draw[->] (Question1.north) |- (Prover1);
\draw[->] (QuestionK.north) |- (ProverK);

\draw[->] (Prover1.east) -| (Answer1);
\draw[->] (ProverK.east) -| (AnswerK);

\draw[->] (Answer1.east) -| (FunctionPi);
\draw (AnswerK.east) -| (DotsBeginRight);
\draw[loosely dotted] (DotsBeginRight) -- (DotsEndRight);
\draw[->] (DotsEndRight) -- (FunctionPi);

\draw[->] (FunctionPi.east) -| (Output);

\draw[line width = 2pt,loosely dotted] (MidDotsBegin) -- (MidDotsEnd);

\end{tikzpicture}
}%
\caption{A $k$-prover one-round proof system simply extends the 2-prover model.}
\label{fig:verifier-kprover}
\end{figure}

Based on this, Feige then provides the following lemma:

\begin{lemma}
\label{lem:3sat-5gap}
Consider a 3CNF-5 formula $\varphi$.
\begin{itemize}
\item If $\varphi$ is satisfiable, then there exists a strategy of the provers that always strongly accepts.
\item If at most a $(1 - \varepsilon)$ fraction of clauses of $\varphi$ are simultaneously satisfiable, then the verifier weakly accepts with probability at most $k^2 \cdot 2^{-c\rho}$, where $k$ is the number of provers and $c$ is a positive constant depending only on $\varepsilon$.
\end{itemize}
\end{lemma}

The error term in the second case is due to pairwise invocation of Lemma \ref{lem:razrepetition}.

\subsubsection{Partition System}

Feige also defines a \emph{partition system}, a generalization of the special set system in Definition \ref{def:specialset} and an invocation of the $(n, k, b)$-anti-universal set of Definition \ref{def:nkbantiuniversalset}.

\begin{definition}[Partition System]
\label{def:partitionsystem}
A \emph{partition system} $\beta(m, \Ell, k, d)$ has the following properties:
\begin{itemize}
\item There is a universe set $\beta$, with $|\beta| = m$.
\item There is a collection of $\Ell$ distinct partitions of $\beta$, $p_1, \dots, p_\Ell$.
\item For $1 \leq i \leq \Ell$, partition $p_i$ is a collection of $k$ disjoint subsets of $\beta$ whose union is $\beta$.
\item Any cover of $\beta$ by subsets such that no two subsets are from the same partition requires at least $d$ subsets.
\end{itemize}
\end{definition}

Feige obtains the following results:

\begin{lemma}
For every $c \geq 0$ and sufficiently large $m$, there exists a partition system $\beta (m, \Ell, k, d)$ such that all of the following hold:
\begin{itemize}
\item $\Ell \simeq (\log m)^c$.
\item $k$ may be chosen arbitrarily as long as $k < \ln \frac{m}{3} \ln \ln m$.
\item $d = (1 - f(k)) k \ln m$, where $\lim\limits_{k \to \infty} f(k) = 0$.
\end{itemize}
Moreover, such a partition system can be constructed in $\text{ZTIME}(m^{\mathcal{O}(\log m)})$. By Theorem \ref{thm:nkbantiuniversalsetconstruction}, there also exists a deterministic construction taking time linear in $m$ and satisfies $m = \left( \frac{k}{k - 1} \right)^d d^{\mathcal{O}(\log d)} \log \Ell$ for arbitrary constant $k$.
\end{lemma}

\subsection{The Proof}

Because the core lemma of Feige's paper relies on the construction, we defer presenting it until after the initial setup.

\subsubsection{Initial Setup}

Consider a $k$-prover system for MAX 3SAT-5. The verifier's randomness is represented by a seed bitstring $r$ of length $\rho \log 5n$. There are $R = (5n)^{\rho}$ such seeds. Each seed induces a selection of $\rho$ clauses and $\rho$ distinguished variables, one from each of the clauses. Next, each seed $r \in R$ is associated with a distinct partition system $\beta_r (m, \Ell, k, d)$ with the following properties:

\begin{itemize}
\item $\Ell = 2^{\rho}$
\item $m = n^{\Theta(\rho)}$
\item $d = (1 - f(k))k \ln m$
\end{itemize}

Note that the corresponding universe of our \setcover{} instance will be $\mathcal{B} = \bigcup_r \beta_r$. This makes the size of the universe $N = mR = n^{\Theta(\rho)} (5n)^\rho$.

Within a given partition system $\beta_r$, each of the $2^{\rho}$ partitions $p_{r, j}$ is associated with a $\rho$-bit string $s_{r, j}$, representing assignments to the distinguished variables. Each subset within $p_{r, j}$ is associated with a unique prover $i$; denote the subset associated with prover $i$ in partition $j$ of system $\beta_r$ by $\beta_r [j, i]$.

We seek to generate a subset $S_{q, a, i}$ for each question-answer pair $(q, a)$ of prover $P_i$. As before, let $q(r, i)$ be the question sent to prover $P_i$ when using seed $r$. For all $r$ such that $q = q(r, i)$, consider the induced selection of distinguished variables given by the bitstring $r$. From the corresponding answer $a$, consider the induced assignment $a_r$ to the distinguished variables on each coordinate. Because every assignment is represented by some partition in the system $\beta_r$, some partition $p^*_r$ in that system has $a_r$ as its associated string. Consider the subset $\beta_r [p^*_r, i]$. Define

$$S_{q, a, i} = \bigcup_{\{r \mid q = q(r, i)\}} \beta_r [p^*_r, i]$$

or the union of the subset corresponding to prover $P_i$ in the partition corresponding to the answers of $a$, across all partition systems whose seed asks question $q$. This generates the sets of our \setcover{} instance. A diagram of the dependencies in this construction can be found in Figure \ref{fig:construction-feige}.

\begin{figure}[tb!]
\centering
\begin{tikzpicture}[node distance=1.5cm, font=\small,thick, >=Latex]

\node[goalconstruction]
	(SetCover)
	{\textbf{Set Cover Instance}};

\node[goalelement, below of = SetCover]
	(SetCoverU)
	{Universe Set};

\node[goalelement, below of = SetCoverU]
	(SetCoverS)
	{Subsets};

\node[construction, below right of = SetCover, xshift=4.25cm, yshift=0.5cm]
	(Partitions)
	{\textbf{Collection of Partition Systems}};

\node[element, below of = Partitions]
	(PartitionsBeta)
	{Universe Sets};

\node[element, below of = PartitionsBeta]
	(PartitionsPi)
	{Partitions};

\node[element, below of = PartitionsPi]
	(PartitionSets)
	{Sets in Partitions};

\node[construction, below left of = SetCover, xshift=-4cm, yshift=0.5cm] 
	(Verifier)
	{\textbf{Verifier for MAX 3SAT-5}};

\node[element, below of = Verifier]
	(VerifierR)
	{Random Seeds};

\node[element, below of = VerifierR]
	(VerifierQ)
	{Question Sets};

\node[element, below of = VerifierQ]
	(VerifierA)
	{Answer Sets};

\node[element, below of = VerifierA]
	(VerifierP)
	{Provers};

\path (VerifierQ) -- (VerifierA) node[coordinate,midway] (Dummy1Helper) {};

\node [coordinate, right of = Dummy1Helper, xshift = 1cm] (Dummy1) {};
\node [coordinate, right of = VerifierR, xshift = 1cm] (Dummy2) {};
\node [coordinate, left of = PartitionsBeta, xshift = -1cm] (Dummy3) {};

\path (VerifierA) -- (PartitionSets) node[coordinate,midway] (Dummy4Helper) {};
\path (Dummy4Helper) -- (PartitionSets) node[coordinate, midway] (Dummy4) {};

\node [coordinate, right of = VerifierQ, xshift = 1cm] (Dummy5) {};
\node [coordinate, left of = PartitionsPi, xshift = -1cm] (Dummy6) {};

\draw (VerifierR.east) -- (Dummy2);
\draw[->] (Dummy2) |- (SetCoverU);

\draw (PartitionsBeta.west) -- (Dummy3);
\draw[->] (Dummy3) |- (SetCoverU);

\draw[->] (VerifierP) -| (PartitionSets);
\draw (VerifierA) -| (Dummy4);
\draw[->] (Dummy4) -- (PartitionsPi.south west);
\draw[->] (PartitionSets) -- (PartitionsPi);

\draw (VerifierQ.east) -- (Dummy5);
\draw[->] (Dummy5) |- (SetCoverS);

\draw (PartitionsPi.west) -- (Dummy6);
\draw[->] (Dummy6) |- (SetCoverS);

\end{tikzpicture}
\caption{A diagram of Feige's proof structure.}
\label{fig:construction-feige}
\end{figure}
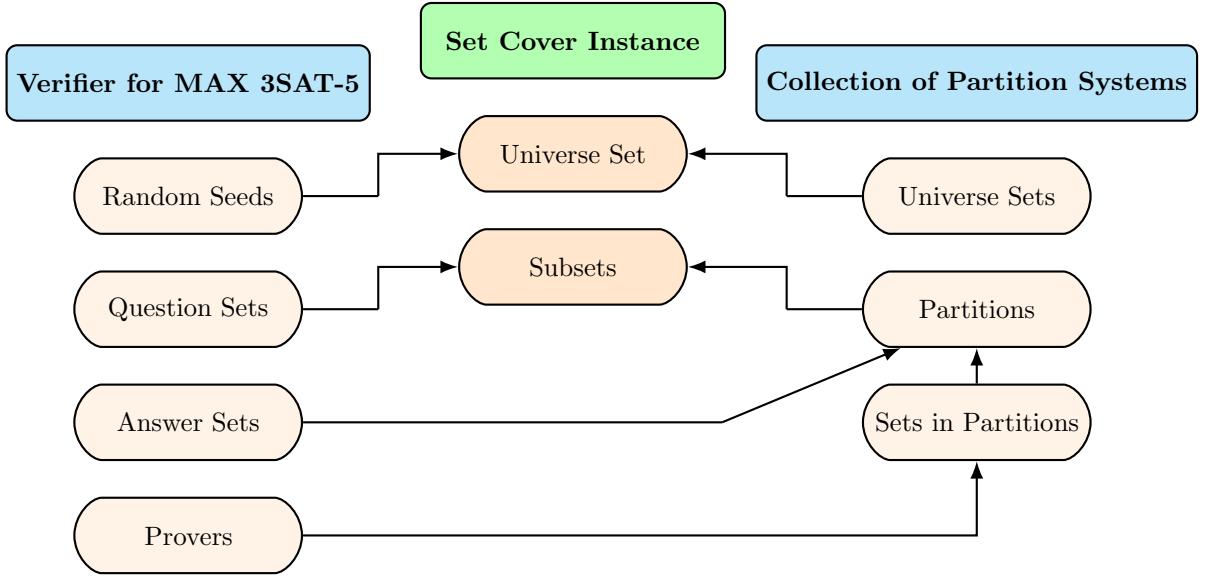

Finally, observe that each prover can be asked a total of $Q = n^{\rho / 2} \cdot \left( \frac{5n}{3} \right)^{\rho / 2}$ questions. We now state the core lemma:

\begin{lemma}
\label{lem:mainclaimfeige}
Let $\varphi$ be a 3CNF-5 formula.
\begin{itemize}
\item If $\varphi$ is satisfiable, then the universe set $\mathcal{B}$ of $N = mR$ points can be covered by $kQ$ subsets.
\item If at most a $(1 - \varepsilon)$ fraction of the clauses of $\varphi$ are simultaneously satisfiable, then $\mathcal{B}$ requires $(1 - 2f(k))kQ \ln m$ subsets to be covered, where $\lim\limits_{k \to \infty} f(k) = 0$.
\end{itemize}
\end{lemma}

We will prove each part of Lemma \ref{lem:mainclaimfeige} separately.

\subsubsection{Proving the First Part}

The first part is straighforward. Consider a satisfiable 3CNF-5 formula $\varphi$. Then by Lemma \ref{lem:3sat-5gap}, there is a strategy that always strongly accepts; in particular, it is one that corresponds to a satisfying assignment of $\varphi$. Fix the provers' strategy to be consistent with a satisfying assignment $a$ for $\varphi$. Consider the sets $S_{q, a, i}$ for which prover $P_i$ replies to question $q$ with the satisfying assignment $a$. Fix a seed $r$ and, from among the considered sets, look at the sets $S_{q_i, a_i, i}$ for all $1 \leq i \leq k$, corresponding to the questions induced by seed $r$ and the answers in alignment with the satisfying assignment for each prover. The universe $\mathcal{B}$ is completely covered by these $k$ sets because, for the partition $p^*_r$ corresponding to $a$, the subset $\beta_r [p^*_r, i]$ is in $S_{q_i, a_i, i}$ for every $i$ by construction, and their union thus completes the partition. 

This argument applies to every seed $r$. This means that each seed's portion is coverable by $k$ sets. By the uniformity of the system, each question $q$ must be askable, so each question must have a corresponding satisfying strategy; this means that $kQ$ sets are used in the covering.

\subsubsection{Proving the Second Part}

Consider now a 3CNF-5 formula $\varphi$ with only a $(1 - \varepsilon)$ fraction of the clauses simultaneously satisfiable. Then by Lemma \ref{lem:3sat-5gap}, any strategy of the provers weakly accepts with probability at most $k^2 \cdot 2^{-c\rho}$. Set $\delta = 2f(k)$ and assume by way of contradiction that there exists a cover $\mathcal{C}$ of size $|\mathcal{C}| \leq (1 - \delta)kQ \ln m$. As in the proof of Lemma \ref{lem:mainclaimly}, we will use this to define a cost function $c$ over seeds:

\begin{itemize}
\item Let $q$ be a question asked to prover $P_i$. The cost of that question $c(q, i)$ is the number of sets $S_{q, a, i}$ in $\mathcal{C}$ for which the query is $q$.
\item The cost of a seed $c(r)$ is $\sum_{i = 1}^{k} c(q(r, i), i)$, or the sum of all individual costs of the $k$ questions induced by that seed.
\end{itemize}

Let a seed $r$ be called \emph{good} if $c(r) <\left(1 - \frac{\delta}{2} \right) k \ln m$, and let it be called \emph{bad} otherwise. Unlike the proof of Lund and Yannakakis for Lemma \ref{lem:mainclaimly}, where we began with deducing an acceptance probability and then reasoned about the size of the covering set, here we assume an upper bound on the size of the covering set and then are able to deduce a lower bound on the fraction of good seeds. We will then use this to prove a lemma on the acceptance probability.

\begin{lemma}
\label{lem:goodseeds}
The fraction of good seeds is at least $\frac{\delta}{2}$.
\end{lemma}

\begin{proof}[Proof of Lemma \ref{lem:goodseeds}] Assume not. Then consider the sum of costs of all seeds. Note that in this case, $\left( 1 - \frac{\delta}{2} \right)$ is the lower bound on the fraction of bad seeds. This yields

$$\sum_{r \in R} c(r) \geq \left( 1 - \frac{\delta}{2} \right) R \cdot \left( 1 - \frac{\delta}{2} \right) k \ln m = \left( 1 - \frac{\delta}{2} \right)^2 kR \ln m$$

Note that $\left( 1 - \frac{\delta}{2} \right)^2 > (1 - \delta)$, resulting in

$$\sum_{r \in R} c(r) > (1 - \delta) kR \ln m$$

However, consider that we can also sum the costs of all seeds by adding the costs of each question asked to each prover by each seed, as in

$$\sum_{r \in R} c(r) = \sum_{r \in R} \sum_{i = 1}^{k} c(q(r, i), i)$$

Note that by uniformity of the proof system, each question is asked $\frac{R}{Q}$ times across the whole seed set. Also, by construction, the sum of weights of all questions across all provers must be $|\mathcal{C}|$. Combining these facts yields

$$\sum_{r \in R} c(r) = \sum_{r \in R} \sum_{i = 1}^{k} c(q(r, i), i) = \frac{R}{Q} \sum_{(q, i)} c(q, i) = \frac{R}{Q} |\mathcal{C}|$$

Substituting this equality into the result of the first sum gives

$$|\mathcal{C}| > (1 - \delta) kR \ln m \geq |\mathcal{C}|$$

which is a contradiction. 
\end{proof}

\begin{lemma}
\label{lem:weakaccept}
Let $\mathcal{C}$ be a covering of $\mathcal{B}$ with $|\mathcal{C}| \leq (1 - \delta)kQ \ln m$. Then, for some strategy of the provers, the verifier weakly accepts $\varphi$ with probability at least $\frac{2 \delta}{(k \ln m)^2}$.
\end{lemma}

\begin{proof}
Consider a randomized strategy of the provers where, on question $q$, prover $P_i$ selects an answer uniformly randomly from the answers $a$ whose sets $S_{q, a, i}$ are in the covering $|\mathcal{C}|$. Note that, for a fixed seed $r$, the corresponding partition system $\beta_r$ admits a bijection between the sets $\beta_r [i, j]$ participating in the cover of $\beta_r$, and sets $S_{q, a, i}$ in the covering $\mathcal{C}$.

Next, consider a good seed $r$. By the property of the partition system, any covering that does not use two subsets from the same partition requires at least $d = (1 - f(k)) k \ln m$ subsets. Recall that $\delta = 2f(k)$ by definition. Since $r$ is good, its partition system $\beta_r$ is covered by $c(r) < \left( 1 - \frac{\delta}{2} \right) k \ln m = (1 - f(k)) k \ln m$ subsets. This is fewer than $d$, so by contraposition, $\beta_r$ must contain two subsets from the same partition in its covering, namely $\beta_r [p_r^*, i]$ and $\beta_r [p_r^*, j]$ for some provers $P_i$ and $P_j$. By the earlier bijection, these correspond to two subsets $S_{q_i, a_i, i}$ and $S_{q_j, a_j, j}$ in the covering $\mathcal{C}$.

Assume the verifier chooses $r$ as its seed. Then prover $P_i$ gets question $q_i$ and chooses an answer uniformly randomly from the set $A_{r, i}$ of all answers $a$ for which $S_{q_i, a, i}$ is in $\mathcal{C}$. The answer $a_i$ is among these. Prover $P_j$ behaves symmetrically and has a corresponding set $A_{r, j}$. Because $r$ is a good seed, $|A_{r, i}| + |A_{r, j}| < k \ln m$. This means that the joint probability the provers will choose to answer with the specific answers $a_i$ and $a_j$ is at least $\left( \frac{1}{(k\ln m) / 2)} \right)^2 = \frac{4}{(k \ln m)^2}$. These answers are consistent, so the verifier will weakly accept if they are chosen.

By Lemma \ref{lem:goodseeds}, the probability $r$ is good is at least $\frac{\delta}{2}$. So the strategy weakly accepts with a probability at least $\frac{\delta}{2} \cdot \frac{4}{(k \ln m)^2} = \frac{2 \delta}{(k \ln m)^2}$. This strategy can be made deterministic by fixing the optimal coin tosses in the randomized strategy, which does not change this lower bound.
\end{proof}

To complete the proof of the second bullet of Lemma \ref{lem:mainclaimfeige}, let $\rho = \Theta(\log \log n)$ and $m = n^{\Theta (\rho)}$. Then

$$ \frac{2\delta}{(k \ln m)^2} > k^2 \cdot 2^{-c \rho}$$

which is a contradiction as the lower bound for weak acceptance probability is higher than the upper bound. Therefore, the size of the cover cannot be as stated and it must be strictly larger than $(1 - \delta)kQ \ln m$. This proves Lemma $\ref{lem:mainclaimfeige}$ and thus Theorem \ref{thm:inapproxfeige}.

\section{Moshkovitz, 2015 \cite{Moshkovitz-2015}}

Moshkovitz further generalizes the framework of Feige to produce even more refined results. She shifts the underlying framework from multi-prover interactive systems to \emph{probabilistically checkable proofs} (PCP), for which many new analysis techniques had been developed in the decade since Feige's lower bound. Notably, PCPs are capable of being analyzed in terms of projection games, which is a framework that lends itself well to \setcover{}. She also begins from the more general \textsc{Label Cover} problem rather than MAX 3SAT-5, which enables an alternative model of \setcover{}. The partition system is exchanged for a bipartite graph coloring with unique properties. Finally, she introduces a modified variant of Conjecture \ref{conj:slidingscale}, the \emph{Projection Games Conjecture}, which the new result is contingent on a weakened version of. These allow her to achieve the following result:

\begin{theorem}
\label{thm:inapproxmoshkovitz}
Unless $\text{P} = \text{NP}$ and assuming the Projection Games Conjecture holds, \setcover{} cannot be approximated within a ratio of $c \ln n$ for any positive $c < 1$.
\end{theorem}

We will proceed to give a proof of this result following the constructions given by Moshkovitz.

\subsection{Preliminary Constructions}

In this section, we describe constructions and concepts needed for the proof. There are many parallels to prior constructions.

\subsubsection{Projection Games}

\begin{definition}[Projection Game]
A \emph{projection game} takes the following inputs:
\begin{itemize}
\item A bipartite graph $G = ((A, B), E)$.
\item Finite alphabets $\Sigma_A, \Sigma_B$.
\item A set $\Phi$ containing a function $\pi_e : \Sigma_A \to \Sigma_B$ for every edge $e \in E$. These functions $\pi_e$ are the \emph{projections}.
\end{itemize}
The \emph{size} of the game is defined as $|A| + |B| + |E|$. The following two concepts are also introduced:
\begin{itemize}
\item A \emph{labeling} is a function mapping a vertex set to its corresponding alphabet, either $\phi_A : A \to \Sigma_A$ or $\phi_B : B \to \Sigma_B$. 
\item A pair of labelings $\phi_A$ and $\phi_B$ \emph{satisfies} edge $e = (a, b)$ if $\pi_e (\phi_A (a)) = \phi_B (b)$; that is, if the projection $\pi_e$ takes the label of $a$ to the label of $b$.
\end{itemize}
The goal of the game is to find labelings $\phi_A$ and $\phi_B$ that satisfy as many edges as possible. An example of a small game is given in Figure \ref{fig:projectiongamesexample}.
\end{definition}

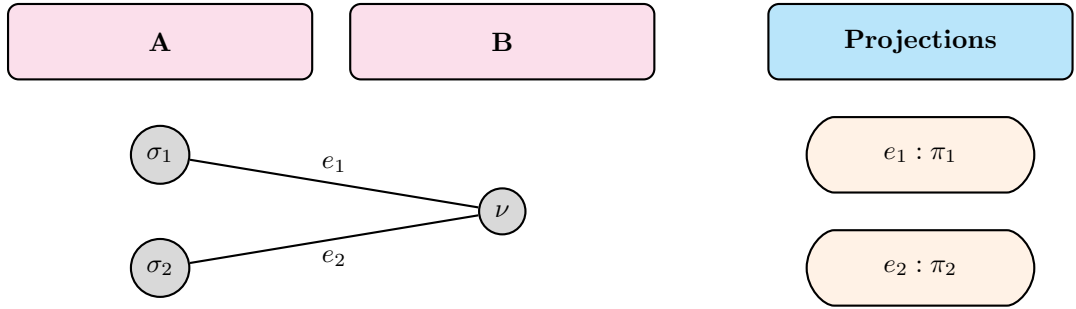
\begin{figure}[tb!]
\centering
\begin{tikzpicture}[node distance=1.5cm, font=\small,thick, >=Latex]

\node[inputconstruction]
	(SetA)
	{$\mathbf{A}$};

\node[inputconstruction, right of = SetA, xshift = 3cm]
	(SetB)
	{$\mathbf{B}$};

\node[construction, right of = SetB, xshift = 4cm]
	(Projections)
	{\textbf{Projections}};

\node[graphnode, below of = SetA]
	(A1)
	{$\sigma_{1}$};

\node[graphnode, below of = A1]
	(A2)
	{$\sigma_{2}$};

\node[graphnode, below of = SetB, yshift = -0.75cm]
	(B)
	{$\nu$};

\node[element, below of = Projections]
	(Pi1)
	{$e_1: \pi_1$};

\node[element, below of = Pi1]
	(Pi2)
	{$e_2: \pi_2$};

\draw (A1) -- (B) node[midway, yshift = 0.25cm] (LabelE1) {$e_1$};
\draw (A2) -- (B) node[midway, yshift = -0.25cm] (LabelE2) {$e_2$};

\end{tikzpicture}
\caption{A small projection game with labelings $\sigma_1, \sigma_2 \in \Sigma_A$ and $\nu \in \Sigma_B$. The goal is to try to choose labels such that $\pi_{1} (\sigma_1) = \pi_{2} (\sigma_2) = \nu$.}
\label{fig:projectiongamesexample}
\end{figure}

The following theorem is due to Moshkovitz:

\begin{theorem}
\label{thm:projectiongamehardness}
Let $\mathcal{P}$ be a projection game of size $n$ with alphabets of size $k$ and soundness error $\varepsilon$, where $k$ and $\varepsilon$ may be functions of $n$. Then, it is NP-hard to distinguish between the case where all edges can be satisfied, and the case where at most a $\varepsilon$ fraction of edges can be satisfied.
\end{theorem}

Unlike the result of Feige, here, the error is explicitly related to the soundness error of the underlying projection game.

\subsubsection{The Projection Games Conjecture}

In \cite{MR-2008}, Moshkovitz and Raz proved the following theorem

\begin{theorem}[Almost-Linear-Size PCP Theorem]
There exists $c > 0$ such that for every $\varepsilon \geq xN^{-c}$, SAT with an input size of $n$ can be reduced to a projection game of input size $N = n^{1 + o(1)} \cdot \text{poly}\left( \frac{1}{\varepsilon} \right)$ over an alphabet of size $\text{exp}\left( \frac{1}{\varepsilon} \right)$ and soundness error $\varepsilon$.
\end{theorem}

In \cite{Moshkovitz-2015}, Moshkovitz adapts the Sliding Scale Conjecture of Bellare \emph{et al.} (Conjecture \ref{conj:slidingscale}) to this framework, and conjectures that the alphabet size can be made polynomial in $\frac{1}{\varepsilon}$ rather than exponential. The formal statement of the conjecture is as follows:

\begin{conjecture}[Projection Games Conjecture]
\label{conj:projectiongames}
There exists $c > 0$ such that for every $\varepsilon \geq N^{-c}$, SAT with an input size of $n$ can be reduced to a projection game of input size $N = n^{1 + o(1)} \cdot \text{poly}\left( \frac{1}{\varepsilon} \right)$ over an alphabet of size $\text{poly}\left( \frac{1}{\varepsilon} \right)$ and soundness error $\varepsilon$.
\end{conjecture}

In both of the above, we may assume without loss of generality that the projection game is \emph{bi-regular}, meaning that all $a \in A$ share a common degree and so do all $b \in B$, though they need not be the same degree for both $A$ and $B$.

We will ultimately formulate the core lemma of Moshkovitz in terms of this conjecture. The full form of the conjecture is not needed to prove that \setcover{} is inapproximable; rather, only a certain weaker version of it is necessary. This version in particular was proven by Dinur and Steurer in \cite{DS-2014}, which leads to the tightest possible lower bound of $ln(n)$ and concludes development of the bound.

\subsubsection{Agreement}

Let $G = ((A, B), E)$ and assume that $(G, \Sigma_A, \Sigma_B, \Phi)$ is a projection game for deciding SAT. Let $\phi_A$ be a labeling of $A$. We say that $A$'s vertices \emph{totally disagree} on a vertex $b \in B$ if, for all pairs of neighbors $a_1, a_2 \in A$ of $b$, the edges $e_1 = (a_1, b)$ and $e_2 = (a_2, b)$ are associated with projections $\pi_1$ and $\pi_2$ for which $\pi_1 (\phi_A (a_1)) \neq \pi_2 (\phi_A (a_2))$. In effect, this means that every pair of neighbors of $b$ is connected by edges whose projections map those neighbors' labels to different elements of $\Sigma_b$ (this is independent of whether or not they agree with the labeling of $b$). This induces the following definition:

\begin{definition}[Agreement Soundness Error]
Consider an unsatisfiable formula $\varphi$ and the projection game $(G, \Sigma_A, \Sigma_B, \Phi)$ as described above. The graph $G$ has \emph{agreement soundness error} $\varepsilon$ if, for any labeling $\phi_A$, $A$'s vertices totally disagree on at least a $(1 - \varepsilon)$ fraction of the $b \in B$.
\end{definition}

Moshkovitz presents two lemmas about agreement. The first asserts that there is a way to construct bi-regular graphs such that any fine-enough partition of the first vertex set forces most vertices of the second vertex set to be connected to many different subsets of the partition. The second connects agreement soundness to the overall soundness of a projection game.

\begin{lemma}
\label{lem:biregular}
For $0 < \varepsilon < 1$, for infinitely many $n$ and $D$, there is an explicit construction of a bi-regular graph $H = ((U, V), E)$ with $|U| = n$, $V$-degree $D$, and $|V| \leq n^{\mathcal{O}(1)}$ such that the following property holds:
\begin{itemize}
\item For every partition of $U$ into subsets $U_1, \dots, U_{m}$ such that $|U_i| \leq \varepsilon |U|$ for all $i$, the fraction of vertices $v \in V$ with more than one neighbor in any $U_i$ is at most $\varepsilon D^2$.
\end{itemize}
\end{lemma}

\begin{lemma}
\label{lem:soundnesscomparison}
Let $D \geq 2$ and $\varepsilon > 0$. Given a projection game with soundness error $\varepsilon^2 D^2$, we can construct a projection game with agreement soundness error $2\varepsilon D^2$ and $B$-degree $D$ in its second vertex set. The transformation preserves the alphabets and increases the game size by only a constant factor.
\end{lemma}

\subsubsection{List-Agreement}

There is a generalization of agreement known as list-agreement, which operates under \emph{$\ell$-labelings}, a more general version of labelings that assign $\ell$ labels to each vertex rather than just one. As before, let $G = ((A, B), E)$ and assume that $(G, \Sigma_A, \Sigma_B, \Phi)$ is a projection game for deciding SAT. For $\ell \geq 1$, let $\widehat{\phi_A} : A \to \binom{\Sigma_A}{\ell}$ be an $\ell$-labeling of $A$. We say that $A$'s vertices \emph{totally list-disagree} on a vertex $b \in B$ if, for all pairs of neighbors $a_1, a_2 \in A$ of $b$, the edges $e_1 = (a_1, b)$ and $e_2 = (a_2, b)$ are associated with projections $\pi_1$ and $\pi_2$ such that, for all pairs of labels $\sigma_1 \in \widehat{\phi_A} (a_1)$ and $\sigma_2 \in \widehat{\phi_A} (a_2)$, $\pi_1 (\sigma_1) \neq \pi_2 (\sigma_2)$. In the case of $\ell = 1$, this is exactly the original notion of total disagreement; in general, it is the statement that the images of $\pi_1$ on the labels of $a_1$ and $\pi_2$ on the labels of $a_2$ are disjoint. As before, this induces a corresponding notion of soundness:

\begin{definition}[List-Agreement Soundness Error]
Consider an unsatisfiable formula $\varphi$ and the projection game $(G, \Sigma_A, \Sigma_B, \Phi)$ as described above. The graph $G$ has \emph{list-agreement soundness error} $(\ell, \varepsilon)$ if, for any $\ell$-labeling $\widehat{\phi_A}$, $A$'s vertices totally list-disagree on at least a $(1 - \varepsilon)$ fraction of the $b \in B$.
\end{definition}

There is a relationship between the agreement soundness error of a given projection game and the corresponding list-agreement soundness error, codified in the following lemma. The proof is short, so we present the proof as well.

\begin{lemma}
\label{lem:agreementsoundnesscomparison}
Let $\ell \geq 1$ and $0 < \varepsilon < 1$. A projection game with agreement soundness error $\varepsilon$ has list-agreement soundness error $(\ell, \varepsilon \ell^2)$.
\end{lemma}

\begin{proof}[Proof of Lemma \ref{lem:agreementsoundnesscomparison}]
Assume not. Then there exists some $\ell$-labeling $\widehat{\phi_A} : A \to \binom{\Sigma_A}{\ell}$ such that, on a fraction of of $B$'s vertices strictly larger than $\varepsilon \ell^2$, $A$'s vertices do not totally list-disagree.

Define a labeling $\phi_A : A \to \Sigma_a$ by assigning each vertex $a \in A$ one of the labels in $\widehat{\phi_A} (a)$ at random. If some $b \in B$ has neighbors $a_1, a_2$ that do not list-disagree on $b$ under $\widehat{\phi_A}$, then they agree on $b$ under $\phi_A$ with probability at least $\frac{1}{\ell^2}$. Therefore, the expected fraction $b \in B$ that has at least two neighbors in agreement is strictly larger than $\varepsilon \ell^2 \cdot \frac{1}{\ell^2} = \varepsilon$. However, this fraction is also $\varepsilon$ by assumption. Contradiction as this implies $\varepsilon > \varepsilon$.
\end{proof}

This progressive relationship between the three notions of soundness will allow us to state and prove the core lemma of Moshkovitz's proof.

\subsection{The Proof}

In terms of projection games, Moshkovitz seeks to prove the following lemma:

\begin{lemma}
\label{lem:mainclaimmoshkovitz}
For every $0 < \alpha < 1$, there exists $c$ a function of $\alpha$ such that, if the Projection Games Conjecture holds with soundness error $\varepsilon = \frac{c}{(\log_2 n)^4}$, then it is NP-hard to approximate Set Cover on inputs of size $N = n^{\mathcal{O}(1 / \alpha)}$ better than $(1 - \alpha) \ln N$.
\end{lemma}

This lemma would, contingent on the Projection Games Conjecture, close the gap between the algorithm of Johnson given in \cite{Johnson-1974} and the result of Feige. The following corollary comes from progressively applying Lemmas \ref{lem:soundnesscomparison} and \ref{lem:agreementsoundnesscomparison} to Conjecture \ref{conj:projectiongames}:

\begin{corollary}
\label{cor:keycorollarymoshkovitz}
For any $\ell = \text{polylog}(n)$, for any constant prime power $D \geq 2$ and constant $0 < \alpha < 1$, SAT on input size $n$ can be reduced to a projection game of size $N = \text{poly}(n)$ with alphabet size $\text{poly}(n)$ and $B$-degree $D$, and where the list-agreement soundness error is $(\ell, \alpha)$.
\end{corollary}

With this corollary, we are able to set up and prove Lemma \ref{lem:mainclaimmoshkovitz}.

\subsubsection{Initial Setup}

Let $\mathcal{G}$ be a projection game for SAT as in Corollary \ref{cor:keycorollarymoshkovitz}. Assume $\mathcal{G}$ is bi-regular with $A$-degree $D_A$ and $B$-degree $D_B$. We will construct a partition system $\beta (m, \Ell, k, d)$ as in Definition \ref{def:partitionsystem} corresponding to $\mathcal{G}$. Recall that this construction is deterministic and efficient by Theorem \ref{thm:nkbantiuniversalsetconstruction}. Additionally, unlike Feige, this proof requires just one partition system. The parameters will be as follows:

\begin{itemize}
\item The universe set will be $\beta$, with $|\beta| = m \leq \text{poly}(k^{\log k}, \log \ell)$.
\item The number of partitions will be $\Ell = |\Sigma_B|$.
\item The number of sets per partition will be $k = D_B$.
\item The minimum set requirement $d$ will be $k \ln m (1 - \alpha)$ for $\alpha \leq \frac{2}{k}$.
\end{itemize}

Next, associate each partition of $\beta$ with a character in $\Sigma_B$, as in $P_{\sigma_1}, \dots, P_{\sigma_\Ell}$. Associate each set within each partition with an index in $[k]$. Similarly, for each vertex $b \in B$, associate each incident edge $e = (a, b)$ with an index in $[k]$.

We will construct a \setcover{} instance $\mathcal{S}$. The universe will be $\beta \times B$. We will include a subset $S_{a, \sigma}$ for every vertex $a \in A$ and every character $\sigma \in \Sigma_A$. To construct these subsets, do the following for each edge $e = (a, b)$ adjacent to $a$: Look at the index $j_e$ associated with that edge by vertex $b$. Then, take the subset $S_{j_e}$ associated with index $j_e$ in the partition corresponding to the character that is the image of $\sigma$ under the projection $\pi_e$. The subset $S_{a, \sigma}$ can then be defined as

$$S_{a, \sigma} = \bigcup_{e = (a, b) \text{ for some } b} S_{j_e}$$

which completes the construction. A diagram of Moshkovitz's construction can be seen in Figure \ref{fig:construction-moshkovitz}.

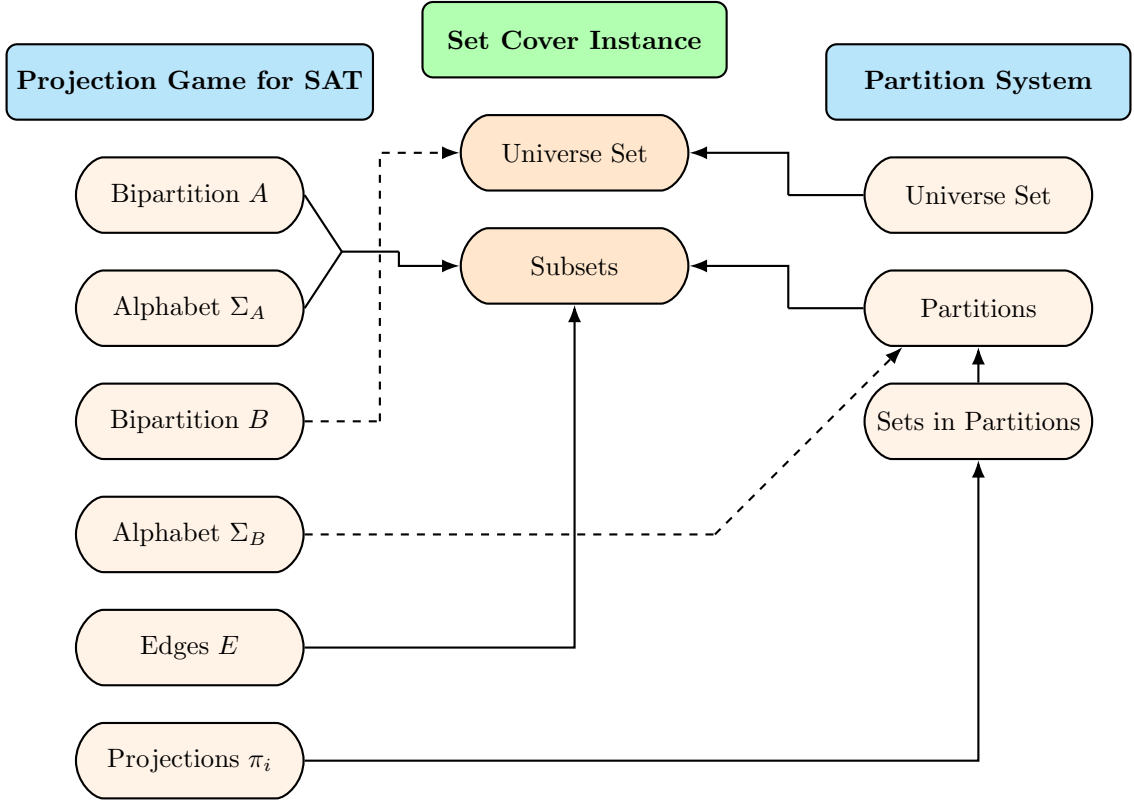
\begin{figure}[tb!]
\centering
\begin{tikzpicture}[node distance=1.5cm, font=\small,thick, >=Latex]

\node[goalconstruction]
	(SetCover)
	{\textbf{Set Cover Instance}};

\node[goalelement, below of = SetCover]
	(SetCoverU)
	{Universe Set};

\node[goalelement, below of = SetCoverU]
	(SetCoverS)
	{Subsets};

\node[construction, below right of = SetCover, xshift=4.25cm, yshift=0.5cm]
	(Partition)
	{\textbf{Partition System}};

\node[element, below of = Partition]
	(PartitionBeta)
	{Universe Set};

\node[element, below of = PartitionBeta]
	(PartitionPi)
	{Partitions};

\node[element, below of = PartitionPi]
	(PartitionSets)
	{Sets in Partitions};

\node[construction, below left of = SetCover, xshift=-4cm, yshift=0.5cm] 
	(ProjGame)
	{\textbf{Projection Game for SAT}};

\node[element, below of = ProjGame]
	(ProjGameGraphA)
	{Bipartition $A$};

\node[element, below of = ProjGameGraphA]
	(ProjGameSigmaA)
	{Alphabet $\Sigma_A$};

\node[element, below of = ProjGameSigmaA]
	(ProjGameGraphB)
	{Bipartition $B$};

\node[element, below of = ProjGameGraphB]
	(ProjGameSigmaB)
	{Alphabet $\Sigma_B$};

\node[element, below of = ProjGameSigmaB]
	(ProjGameEdges)
	{Edges $E$};

\node[element, below of = ProjGameEdges]
	(ProjGameProjections)
	{Projections $\pi_i$};

\path (ProjGameSigmaA) -- (ProjGameGraphA) node[coordinate, midway] (Dummy1Helper) {};
\node [coordinate, right of = Dummy1Helper, xshift = 0.5cm] (Dummy1) {};
\node [coordinate, right of = Dummy1, xshift = -0.75cm] (Dummy1Ext) {};

\node [coordinate, right of = ProjGameGraphB, xshift = 1cm] (Dummy2) {};
\node [coordinate, left of = PartitionBeta, xshift = -1cm] (Dummy3) {};

\node[coordinate, below of = PartitionSets] (Dummy4Helper) {};
\path (Dummy4Helper) -- (ProjGameSigmaB) node[coordinate, midway] (Dummy4Ext) {};
\path (Dummy4Ext) -- (Dummy4Helper) node[coordinate, midway, xshift=-1.25cm] (Dummy4) {};

\node [coordinate, right of = VerifierQ, xshift = 1cm] (Dummy5) {};
\node [coordinate, left of = PartitionsPi, xshift = -1cm] (Dummy6) {};

\draw (ProjGameGraphA.east) -- (Dummy1);
\draw (ProjGameSigmaA.east) -- (Dummy1);
\draw (Dummy1) -- (Dummy1Ext);
\draw[->] (Dummy1Ext) |- (SetCoverS);

\draw[dashed] (ProjGameGraphB.east) -- (Dummy2);
\draw[dashed, ->] (Dummy2) |- (SetCoverU);

\draw (PartitionBeta.west) -- (Dummy3);
\draw[->] (Dummy3) |- (SetCoverU);

\draw[->] (ProjGameProjections) -| (PartitionSets);

\draw[dashed] (ProjGameSigmaB) -| (Dummy4);
\draw[dashed, ->] (Dummy4) -- (PartitionsPi.south west);
\draw[->] (PartitionSets) -- (PartitionsPi);

\draw[->] (ProjGameEdges) -| (SetCoverS.south);

\draw (PartitionsPi.west) -- (Dummy6);
\draw[->] (Dummy6) |- (SetCoverS);

\end{tikzpicture}
\caption{A diagram of Moshkovitz's proof structure.}
\label{fig:construction-moshkovitz}
\end{figure}

\subsubsection{Completeness and Soundness}

We seek to prove that this reduction satisfies completeness and soundness. Completeness is relatively trivial; the statement is as follows:

\begin{lemma}[Completeness]
If all the edges in projection game $\mathcal{G}$ can be satisfied, then $\mathcal{S}$ has a covering of size $|A|$.
\end{lemma}

The proof is simple and is obtained by taking the cover corresponding to the satisfying labeling. The question of soundness is more difficult:

\begin{lemma}[Soundness]
\label{lem:soundness}
Recall $d = k \ln m (1 - \alpha)$. If $\mathcal{G}$ has list-agreement soundness $(d, \alpha)$, then every covering of $\mathcal{S}$ is strictly larger than $|A| \ln m (1 - 2\alpha)$.
\end{lemma}

\begin{proof}[Proof of Lemma \ref{lem:soundness}]
Assume not. Then, there exists a covering $\mathcal{C}$ of size at most $|A| \ln m (1 - 2\alpha)$. Define costs $c(a)$ and $c(b)$ for every $a \in A$ and $b \in B$ as follows:

\begin{itemize}
\item For every $a \in A$, let $c(a)$ be the number of sets $S_{a, \sigma}$ in the covering $\mathcal{C}$. Note that by construction, $\sum_{a \in A} s_a = |\mathcal{C}|$.
\item For every $b \in B$, let $c(b)$ be the number of sets in the covering $\mathcal{C}$ that participate in covering the portion $\beta \times \{b\}$. 
\end{itemize}

Observe that

$$\sum_{b \in B} s_b = D_A \sum_{a \in A} s_a$$

by construction. By the upper bound on the covering size, and the fact that $D_A |A| = D_B |B|$ due to bipartiteness, we get

$$\sum_{b \in B} s_b = D_A \sum_{a \in A} s_a \leq D_A |A| \ln m (1 - 2 \alpha) = D_B |B| \ln m (1 - 2 \alpha)$$

This means that, in expectation over all $b \in B$, each universe $\beta \times \{ b \}$ is covered by at most $D_B \ln m (1 - 2 \alpha)$ sets. Recall that the minimum set requirement $d$ of the partition system is equal to $D_B \ln m (1 - alpha)$. By Markov's Inequality, at least an $\alpha$ fraction of $b \in B$ have their corresponding universe $\beta \times \{ b \}$ covered by at most $d$ sets. By the property of the partition system, for these $b$, each must have two incident edges whose projections agree on some character in the labelings of the two paired vertices of $A$. 

Now, consider the $\ell$-labeling $\widehat{\phi_A} : A \to \binom{\Sigma_A}{d}$ where, for every $a \in A$, the label is made of $d$ symbols $\sigma$ for which $S_{a, \sigma}$ is in $\mathcal{C}$, with arbitrary symbols added as padding if there are not enough. Then, for the $\alpha$-fraction of $b$ described above, $A$'s vertices will not totally list-disagree. Contradiction as $\mathcal{G}$ has list-agreement soundness $(d, \alpha)$.
\end{proof}

To complete the reduction, fix a constant $0 < \alpha < 1$. Assume that the Projection Games Conjecture holds for $\varepsilon = \frac{c}{(\log_2 n)^4}$ for some constant $c$ depending only on $\alpha$. This enables us to use Corollary \ref{cor:keycorollarymoshkovitz} to generate a projection game for SAT with list-agreement soundness $(d, \alpha)$. Let $N = m|B|$ be the number of elements in the universe set of the \setcover{} instance $\mathcal{S}$. Take $m = \Theta (|B|^{1 / \alpha})$; this may require duplicating elements. Then $\ln N = (1 + \alpha) \ln m$ with an approximation ratio of at least $(1 - 3\alpha) \ln N$. The entire reduction is polynomial in $n$. This proves Lemma \ref{lem:mainclaimmoshkovitz}.

\section{Dinur and Steurer, 2014 \cite{DS-2014}}

The contribution of Dinur and Steurer is a proof of the weaker version of Conjecture \ref{conj:projectiongames} necessary to achieve the results of Moshkovitz. In particular, they prove the following result about the \textsc{Label Cover} problem:

\begin{theorem}
For every constant $c > 0$, given a \textsc{Label Cover} instance of size $n$ with alphabet size at most $n$, it is NP-hard to decide if its value is 1 or at most $\varepsilon = (\log n)^{-c}$.
\end{theorem}

This is sufficient to prove Lemma \ref{lem:mainclaimmoshkovitz} without reliance on a conjecture, hence proving that \setcover{} is inapproximable within $(1 - \varepsilon) \ln n$ for any $0 < \varepsilon < 1$. We therefore arrive at the current result of Theorem \ref{thm:primary}.

\newpage
\appendix
\section*{Appendix A: Chronological Timelines}
\makeatletter
\def\@currentlabel{A}
\makeatother
\label{app:timeline}
\renewcommand{\thesubsection}{\Roman{subsection}}

\subsection{Chronological Timeline with Results as Stated}
The following timeline presents approximate order of \emph{initial} results, with the years representing the final dates of publication of each paper. The values given are the new lower bounds on approximability of \setcover{}; note that the main paper presents theorems largely in terms of \emph{upper} bounds on \emph{inapproximability}. See Section \ref{sec:timeline} for more.

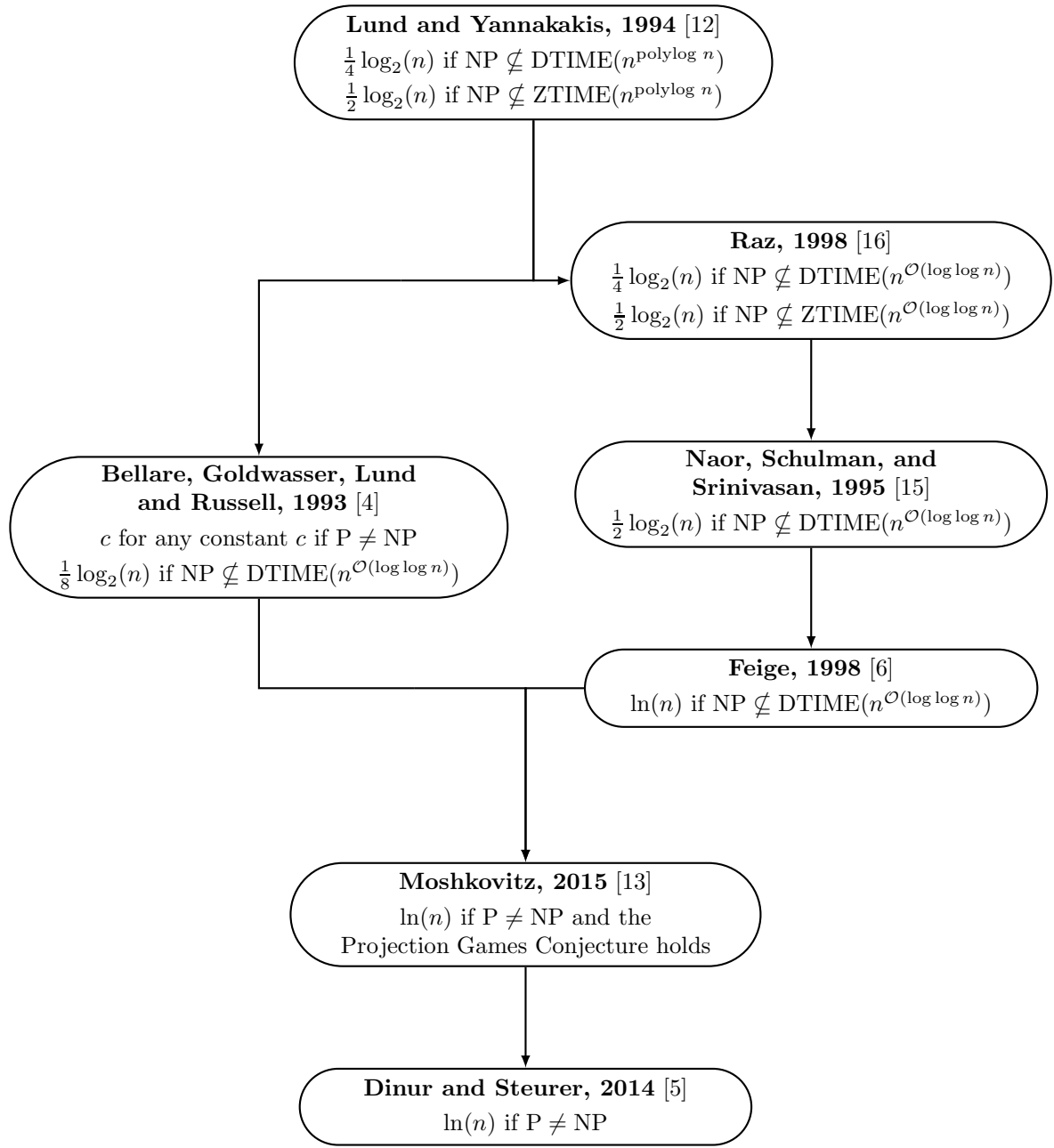
\begin{tikzpicture}[font=\small,thick, >=latex]

\node[draw,
	rounded rectangle,
	minimum width=2.5cm,
	minimum height=1cm] (LY) {\parbox{6cm}{\centering \textbf{Lund and Yannakakis, 1994} \cite{LY-1994} \\[0.3em] $\frac{1}{4} \log_2 (n)$ if $\text{NP} \nsubseteq \text{DTIME}(n^{\text{polylog }n})$ \\[0.3em]  $\frac{1}{2} \log_2 (n)$ if  $\text{NP} \nsubseteq \text{ZTIME}(n^{\text{polylog }n})$}};

\node[draw,
	rounded rectangle,
	below left=5cm and -1.25cm of LY,
	minimum width=2.5cm,
	minimum height=1cm] (BGLR) {\parbox{6cm}{\centering \textbf{Bellare, Goldwasser, Lund} \\ \textbf{and Russell, 1993} \cite{BGLR-1993} \\[0.3em] $c$ for any constant $c$ if $\text{P} \neq \text{NP}$ \\[0.3em] $\frac{1}{8} \log_2 (n)$ if $\text{NP} \nsubseteq \text{DTIME}(n^{\mathcal{O} (\log \log n)})$}};

\node[draw,
	rounded rectangle,
	below right=1.5cm and -1.25cm of LY,
	minimum width=2.5cm,
	minimum height=1cm] (Raz) {\parbox{6cm}{\centering \textbf{Raz, 1998} \cite{Raz-1998} \\[0.3em] $\frac{1}{4} \log_2 (n)$ if $\text{NP} \nsubseteq \text{DTIME}(n^{\mathcal{O} (\log \log n)})$ \\[0.3em]  $\frac{1}{2} \log_2 (n)$ if $\text{NP} \nsubseteq \text{ZTIME}(n^{\mathcal{O} (\log \log n)})$}};

\node[coordinate, left = 2.5cm of Raz] (Dummy) {};

\node[draw,
	rounded rectangle,
	below=1.5cm of Raz,
	minimum width=2.5cm,
	minimum height=1cm] (Naor) {\parbox{6cm}{\centering \textbf{Naor, Schulman, and} \\ \textbf{Srinivasan, 1995} \cite{NSS-1995} \\[0.3em] $\frac{1}{2} \log_2 (n)$ if $\text{NP} \nsubseteq \text{DTIME}(n^{\mathcal{O} (\log \log n)})$}};

\node[draw,
	rounded rectangle,
	below=1.5cm of Naor,
	minimum width=2.5cm,
	minimum height=1cm] (Feige) {\parbox{6cm}{\centering \textbf{Feige, 1998} \cite{Feige-1998} \\[0.3em] $\ln(n)$ if $\text{NP} \nsubseteq \text{DTIME}(n^{\mathcal{O} (\log \log n)})$}};

\node[coordinate, left=2.5cm of Feige] (Dummy2) {};

\node[draw,
	rounded rectangle,
	below left=2cm and -1.25cm of Feige,
	minimum width=2.5cm,
	minimum height=1cm] (Mosh) {\parbox{6cm}{\centering \textbf{Moshkovitz, 2015} \cite{Moshkovitz-2015} \\[0.3em] $\ln(n)$ if $\text{P} \neq \text{NP}$ and the \\ Projection Games Conjecture holds}};

\node[draw,
	rounded rectangle,
	below=1.5cm of Mosh,
	minimum width=2.5cm,
	minimum height=1cm] (DS) {\parbox{6cm}{\centering \textbf{Dinur and Steurer, 2014} \cite{DS-2014} \\[0.3em] $\ln(n)$ if $\text{P} \neq \text{NP}$}};

\draw (LY) |- (Dummy);
\draw[->] (Dummy) -| (BGLR);
\draw[->] (LY) |- (Raz);
\draw[->] (Raz) -- (Naor);
\draw[->] (Naor) -- (Feige);
\draw[->] (Feige) -| (Mosh);
\draw (BGLR) |- (Dummy2);
\draw[->] (Dummy2) -| (Mosh);
\draw[->] (Mosh) -- (DS);

\end{tikzpicture}

\newpage
\subsection{Chronological Timeline in Terms of \protect$\ln(n)$}
The following timeline again presents approximate order of \emph{initial} results, with dates given being actual publication dates. Because logarithms differ from one another by only a constant factor, it is possible to present all results of the timeline in terms of $\ln(n)$. This timeline is presented as a way of clearly observing the improvement in results. Decimals are rounded to the nearest thousandth.
\vspace{0.15in}

\begin{tikzpicture}[font=\small,thick, >=latex]

\node[draw,
	rounded rectangle,
	minimum width=2.5cm,
	minimum height=1cm] (LY) {\parbox{6cm}{\centering \textbf{Lund and Yannakakis, 1994} \cite{LY-1994} \\[0.3em] $0.361 \ln (n)$ if $\text{NP} \nsubseteq \text{DTIME}(n^{\text{polylog }n})$ \\[0.3em]  $0.721 \ln (n)$ if  $\text{NP} \nsubseteq \text{ZTIME}(n^{\text{polylog }n})$}};

\node[draw,
	rounded rectangle,
	below left=5cm and -1.25cm of LY,
	minimum width=2.5cm,
	minimum height=1cm] (BGLR) {\parbox{6.25cm}{\centering \textbf{Bellare, Goldwasser, Lund} \\ \textbf{and Russell, 1993} \cite{BGLR-1993} \\[0.3em] $c$ for any constant $c$ if $\text{P} \neq \text{NP}$ \\[0.3em] $0.180 \ln (n)$ if $\text{NP} \nsubseteq \text{DTIME}(n^{\mathcal{O} (\log \log n)})$}};

\node[draw,
	rounded rectangle,
	below right=1.5cm and -1.25cm of LY,
	minimum width=2.5cm,
	minimum height=1cm] (Raz) {\parbox{6.25cm}{\centering \textbf{Raz, 1998} \cite{Raz-1998} \\[0.3em] $0.361 \ln (n)$ if $\text{NP} \nsubseteq \text{DTIME}(n^{\mathcal{O} (\log \log n)})$ \\[0.3em]  $0.721 \ln (n)$ if $\text{NP} \nsubseteq \text{ZTIME}(n^{\mathcal{O} (\log \log n)})$}};

\node[coordinate, left = 2.5cm of Raz] (Dummy) {};

\node[draw,
	rounded rectangle,
	below=1.5cm of Raz,
	minimum width=2.5cm,
	minimum height=1cm] (Naor) {\parbox{6.25cm}{\centering \textbf{Naor, Schulman, and} \\ \textbf{Srinivasan, 1995} \cite{NSS-1995} \\[0.3em] $0.721 \ln (n)$ if $\text{NP} \nsubseteq \text{DTIME}(n^{\mathcal{O} (\log \log n)})$}};

\node[draw,
	rounded rectangle,
	below=1.5cm of Naor,
	minimum width=2.5cm,
	minimum height=1cm] (Feige) {\parbox{6cm}{\centering \textbf{Feige, 1998} \cite{Feige-1998} \\[0.3em] $\ln(n)$ if $\text{NP} \nsubseteq \text{DTIME}(n^{\mathcal{O} (\log \log n)})$}};

\node[coordinate, left=2.5cm of Feige] (Dummy2) {};

\node[draw,
	rounded rectangle,
	below left=2cm and -1.25cm of Feige,
	minimum width=2.5cm,
	minimum height=1cm] (Mosh) {\parbox{6cm}{\centering \textbf{Moshkovitz, 2015} \cite{Moshkovitz-2015} \\[0.3em] $\ln(n)$ if $\text{P} \neq \text{NP}$ and the \\ Projection Games Conjecture holds}};

\node[draw,
	rounded rectangle,
	below=1.5cm of Mosh,
	minimum width=2.5cm,
	minimum height=1cm] (DS) {\parbox{6cm}{\centering \textbf{Dinur and Steurer, 2014} \cite{DS-2014} \\[0.3em] $\ln(n)$ if $\text{P} \neq \text{NP}$}};

\draw (LY) |- (Dummy);
\draw[->] (Dummy) -| (BGLR);
\draw[->] (LY) |- (Raz);
\draw[->] (Raz) -- (Naor);
\draw[->] (Naor) -- (Feige);
\draw[->] (Feige) -| (Mosh);
\draw (BGLR) |- (Dummy2);
\draw[->] (Dummy2) -| (Mosh);
\draw[->] (Mosh) -- (DS);

\end{tikzpicture}

\newpage





\end{document}